\documentclass[aoas]{imsart}

\RequirePackage{amsthm,amsmath,amsfonts,amssymb}
\RequirePackage[authoryear]{natbib}

\usepackage{varioref}
\RequirePackage[colorlinks,citecolor=blue,urlcolor=blue]{hyperref}%% uncomment this for coloring bibliography citations and linked URLs

\RequirePackage{graphicx}%% uncomment this for including figures

\usepackage{probstat}
\RequirePackage{amsmath,amsfonts,mathtools}
\usepackage{float, graphicx, enumitem, caption}

\usepackage{algorithm}
\usepackage{algpseudocode}

\usepackage{longtable}
\usepackage{array}
\usepackage{multirow}

\usepackage{soul}
\usepackage{xcolor}

\usepackage{varioref}
\usepackage{xr}
\usepackage[colorlinks,citecolor=blue,urlcolor=blue]{hyperref}

\allowdisplaybreaks

\startlocaldefs
\theoremstyle{plain}
\newtheorem{theorem}{Theorem}

\newtheorem{proposition}{Proposition}
\newtheorem{lemma}{Lemma}

\theoremstyle{remark}
\newtheorem{assumption}{Assumption}
\usepackage{apptools}
\AtAppendix{\counterwithin{theorem}{section}}
\AtAppendix{\counterwithin{corollary}{section}}
\AtAppendix{\counterwithin{assumption}{section}}
\AtAppendix{\counterwithin{proposition}{section}}
\AtAppendix{\counterwithin{lemma}{section}}
\AtAppendix{\counterwithin{table}{section}}
\AtAppendix{\counterwithin{figure}{section}}

\endlocaldefs

\begin{document}

\begin{frontmatter}
%%%%%%%%%%%%%%%%%%%%%%%%%%%%%%%%%%%%%%%%%%%%%%
%%                                          %%
%% Enter the title of your article here     %%
%%                                          %%
%%%%%%%%%%%%%%%%%%%%%%%%%%%%%%%%%%%%%%%%%%%%%%
\title{Treatment Effect Heterogeneity and Importance Measures for Multivariate Continuous Treatments}
%\title{A sample article title with some additional note\thanksref{T1}}
\runtitle{Heterogeneous Effects of Multivariate Treatments}
%\thankstext{T1}{A sample of additional note to the title.}

\begin{aug}
%%%%%%%%%%%%%%%%%%%%%%%%%%%%%%%%%%%%%%%%%%%%%%%
%% Only one address is permitted per author. %%
%% Only division, organization and e-mail is %%
%% included in the address.                  %%
%% Additional information can be included in %%
%% the Acknowledgments section if necessary. %%
%% ORCID can be inserted by command:         %%
%% \orcid{0000-0000-0000-0000}               %%
%%%%%%%%%%%%%%%%%%%%%%%%%%%%%%%%%%%%%%%%%%%%%%%
\author[A]{\fnms{Heejun}~\snm{Shin}\ead[label=e1]{heejunshin@hsph.harvard.edu}},
\author[B]{\fnms{Antonio}~\snm{Linero}\ead[label=e2]{antonio.linero@austin.utexas.edu}}
\author[A]{\fnms{Michelle}~\snm{Audirac}\ead[label=e3]{maudirac@hsph.harvard.edu}}
\author[A]{\fnms{Kezia}~\snm{Irene}\ead[label=e4]{kezia\_irene@alumni.harvard.edu}}
\author[A]{\fnms{Danielle}~\snm{Braun}\ead[label=e5]{dbraun@mail.harvard.edu}}
\and
\author[C]{\fnms{Joseph}~\snm{Antonelli}\ead[label=e6]{jantonelli@ufl.edu}}
%%%%%%%%%%%%%%%%%%%%%%%%%%%%%%%%%%%%%%%%%%%%%%
%% Addresses                                %%
%%%%%%%%%%%%%%%%%%%%%%%%%%%%%%%%%%%%%%%%%%%%%%
\address[A]{Department of Biostatistics, Harvard T.H. Chan School of Public Health\printead[presep={,\ }]{e1,e3,e4,e5}}
\address[B]{Department of Statistics and Data Science, University of Texas at Austin\printead[presep={,\ }]{e2}}
\address[C]{Department of Statistics, University of Florida\printead[presep={,\ }]{e6}}
\end{aug}

\begin{abstract}
Estimating the joint effect of a multivariate, continuous exposure is crucial, particularly in environmental health where interest lies in simultaneously evaluating the impact of multiple environmental pollutants on health. We develop novel methodology that addresses two key issues for estimation of treatment effects of multivariate, continuous exposures. We use nonparametric Bayesian methodology that is flexible to ensure our approach can capture a wide range of data generating processes. Additionally, we allow the effect of the exposures to be heterogeneous with respect to covariates. Treatment effect heterogeneity has not been well explored in the causal inference literature for multivariate, continuous exposures, and therefore we introduce novel estimands that summarize the nature and extent of the heterogeneity, and propose estimation procedures for new estimands related to treatment effect heterogeneity. We provide theoretical support for the proposed models in the form of posterior contraction rates and show that it works well in simulated examples both with and without heterogeneity. Our approach is motivated by a study of the health effects of simultaneous exposure to the components of PM$_{2.5}$, where we find that the negative health effects of exposure to environmental pollutants are exacerbated by low socioeconomic status, race and age.
\end{abstract}


\end{frontmatter}
%%%%%%%%%%%%%%%%%%%%%%%%%%%%%%%%%%%%%%%%%%%%%%
%% Please use \tableofcontents for articles %%
%% with 50 pages and more                   %%
%%%%%%%%%%%%%%%%%%%%%%%%%%%%%%%%%%%%%%%%%%%%%%
%\tableofcontents

%%%%%%%%%%%%%%%%%%%%%%%%%%%%%%%%%%%%%%%%%%%%%%
%%%% Main text entry area:
\section{Introduction}\label{sec: introduction}

An important scientific question is understanding the effect of a multivariate treatment on an outcome, particularly in environmental health where individuals are exposed simultaneously to multiple pollutants  and it is of interest to understand the  joint impact of each of these pollutants on public health. Further, the effects of each of these pollutants might be heterogeneous with respect to characteristics of the individuals exposed, and it is important to account for and understand this treatment effect heterogeneity. Ignoring this heterogeneity can lead to misspecified outcome models, which lead to biased estimates of the average effect of the treatments on the outcome, and can also mask substantial effects of the treatment on specific subgroups of the population. In this paper, we address this problem by developing methodology for estimating heterogeneous treatment effects of continuous, multivariate treatments.

In the past decade, there has been an increasing interest in analyzing the effects of multivariate exposures in environmental health where the exposures are referred to as environmental mixtures \citep{dominici2010protecting, agier2016systematic, stafoggia2017statistical, gibson2019complex, shin2023spatial}.
Analyzing environmental mixtures is a challenging problem for a number of reasons. For one, the scientific goal is not simply predicting the outcome, but rather understanding the effects of each of the individual exposures, and whether these exposures interact with each other. Therefore, off-the-shelf modern regression strategies such as tree-based models \citep{breiman2001random, chipman2010bart} and Gaussian processes \citep{banerjee2013efficient} do not immediately apply. Consequently, there has been significant interest in developing methods to adapt these models for causal inference \citep{hill2011bayesian, hahn2020bayesian, ray2020semiparametric}. A number of studies in the environmental statistics literature have tailored these methods for multivariate exposures. A common theme among these methods is the use of nonparametric Bayesian models. Gaussian processes are adapted for this purpose in the popular Bayesian kernel machine regression (BKMR, \cite{bobb2015bayesian}) or in related work that explicitly identifies interactions between exposures \citep{ferrari2020identifying}. Related Bayesian approaches use basis function expansions to identify important exposures or interactions among exposures for environmental mixtures \citep{antonelli2020estimating, wei2020sparse, samanta2022estimation}. Other, related approaches have been developed all with the goal of estimating the health effects of environmental mixtures and identifying exposures that affect the outcome \citep{herring2010nonparametric,carrico2015characterization, narisetty2019selection, boss2021hierarchical, ferrari2021bayesian}. 

While the aforementioned approaches are useful for multivariate, continuous exposures, they all implicitly assume that the effect of the environmental exposures is the same across the entire population. This assumption does not hold in the context of ambient air pollution, where  effects have been shown to vary by characteristics such as age, race, and location \citep{wang2020impact, bargagli2020causal}. Treatment effect heterogeneity has seen an explosion of interest in the causal inference literature, though it is typically restricted to a single treatment, either binary \citep{athey2016recursive, wager2018estimation, hahn2020bayesian, semenova2021debiased, fan2022estimation, shin2023improved} or with multiple levels \citep{chang2024individualized}. In the binary setting, heterogeneity is typically summarized by the conditional average treatment effect (CATE) function, $\E\{Y(1)-Y(0)|\bX = \bx\}$ where $Y(w)$ denotes the potential outcome under a binary exposure level $w$ and $\bX$ denotes covariates. While the CATE is a functional estimand, a number of summaries of this function have been proposed to simplify heterogeneity further. \cite{henderson2020individualized} and \cite{chen2024bayesian} examined the posterior probabilities of differential treatment effect, which measures how much the posterior distributions of each CATE, evaluated at the observed covariates, differ from the sample mean of the CATE. Another straightforward univariate summary is the variance of the CATE (VTE, \cite{levy2021fundamental}), which describes the overall degree of heterogeneity. Related quantities of interest to our study  are treatment effect variable importance measures (TE-VIM), which have been recently proposed for univariate treatments in \cite{hines2022variable} and \cite{li2023targeted} to learn which covariates are key factors driving treatment effect heterogeneity. TE-VIM relates regression-based variable importance measures \citep{zhang2020floodgate, williamson2021nonparametric} and the VTE of \cite{levy2021fundamental} within a causal inference framework. Nearly all of this work on treatment effect heterogeneity focuses on binary treatments, and there has been little to no work on the heterogeneous effects of multivariate, continuous exposures.

There exists a substantial gap in the literature on how to define estimands for continuous, multivariate exposures in the presence of treatment effect heterogeneity, and on how to estimate these quantities using observed data and flexible modeling approaches that make as few modeling assumptions as possible. In this paper, we aim to fill this gap by first proposing new estimands in this setting that are interpretable and help describe the complex effect of multiple treatments, and how this effect varies by observed covariates.  We extend treatment effect variable importance measures to this more complex setting and describe how to identify and estimate these estimands from observed data. These measures provide researchers with insight into causal effect modifiers for multivariate, continuous exposures and can provide policy-makers with detailed information about how treatments impact outcomes. Additionally, we develop novel estimation strategies using nonparametric Bayesian methodology. By decomposing outcome regression surfaces into distinct, identifiable functions, we are able to explicitly shrink treatment effects towards homogeneity, while still allowing for heterogeneity when it exists. We also make limited modeling assumptions as we use novel extensions of Bayesian additive regression trees (BART) such as the SoftBART prior distribution \citep{linero2018bayesian} and targeted smoothing BART \citep{starling2020bart, li2022adaptive} for interactions between exposures and covariates to allow for heterogeneity without imposing strong parametric assumptions. 

In addition to methodological development, our nationwide study of the United States Medicare population provides novel epidemiological insights into the health effects of simultaneous exposure to multiple air pollutants. We find that increased exposure to ambient air pollution leads to increases in mortality rates, and that this effect varies by race, age and socioeconomic status. This is the first study that we are aware of to investigate heterogeneity of air pollution mixture effects, and provides public health researchers with more detailed evidence about the negative health effects of air pollution. 

\section{Estimands of Interest}\label{sec: estimands}

Throughout, we assume that we observe $\boldsymbol{\mathcal{D}}=\{(Y_i,\bX_i,\bW_i)\}_{i=1}^{n}$ for each individual where $Y_i\in\mathcal{Y}\subset\reals$ is the outcome of interest, $\bX_i\in\mathcal{X}\subset\reals^p$ is a vector of pre-treatment covariates, and $\bW_i\in\mathcal{W}\subset\reals^q$ is a vector of continuous exposures. Also, we adopt the potential outcomes framework: $Y_i(\bw)$ denotes the potential outcome that would be observed under exposure $\bw$. In order to denote potential outcomes in this manner, and to identify causal effects from the observed data, we make the following assumptions:
\begin{assumption}[Causal identification assumptions]\hfill
\begin{enumerate}
    \item[(i)] \textit{SUTVA} (Stable Unit Treatment Value Assumption, \cite{rubin1980randomization}): Treatments of one unit do not affect the potential outcomes of other units (\textit{no interference}), and there are not different versions of treatments, such that $Y_i=Y_i(\bW_i)$.
    \item[(ii)] \textit{Positivity}: $0 < f_{W | X}(\bw \vert \bX = \bx)$ for all $\bx$ and $\bw$, where $f_{W | X}$ is the conditional density of the treatments given covariates.
    \item[(iii)] \textit{Unconfoundedness}: $Y(\bw) \independent \bW | \bX$ for all $\bw$. 
    \end{enumerate}
\end{assumption}
The positivity assumption guarantees that each unit has a non-zero probability to be exposed to each treatment level for all possible values of pre-treatment variables, at least in large samples. While in principle this assumption is empirically verifiable, there is not a well-established approach to doing so for multivariate, continuous treatments. Though not the focus of our work here, in Appendix \ref{sec:Appendix_positivity} we outline two distinct approaches to assess the plausibility of the multivariate positivity assumption and to examine how robust our results are to this critical assumption. First, we examine the marginal distribution of each exposure given the covariates and check whether each observation satisfies a univariate positivity criterion for each of the exposures. Second, we consider a trimmed causal estimand, which restricts the population being targeted to those observations with higher values of the conditional density of the treatments given the covariates when evaluated at the treatment values of interest. Unconfoundedness requires that the treatment $\bW$ is independent of the potential outcomes, and it implies that there exist no unmeasured variables that confound the treatment-outcome relationship.

Because exposures and covariates are both multivariate, there are many causal estimands one could define. In the following sections, we first define estimands targeting both average treatment effects and treatment effect heterogeneity, as well as variable importance measures summarizing this heterogeneity. While our focus in this paper remains on multivariate continuous exposures, we note that the estimands we propose can be applied to a broader class of exposures with multiple levels, regardless of their dimensionality. In such cases, estimation complexity is significantly reduced due to the lower dimensionality of univariate exposures or the fewer possible interventions for discrete exposures.

\subsection{Marginal and Heterogeneous Effects Estimands}

One potential estimand to look at in this scenario is $\E[Y(\bw) | \bX = \bx]$, however, there are infinitely many values that $\bw$ can take and therefore it is difficult to interpret this estimand. One option is to look at values of $\bw$ where all but one of the exposures is fixed, and then visualize this estimand as a function of one exposure. A simpler setting is one in which there are two exposure levels $\bw_1$ and $\bw_0$, under which we want to compare the outcomes. For example, one may be interested in examining what would happen if an intervention was applied that lowered the level of pollution by a specific amount. In this case, one can use the pollution levels without the intervention as $\bw_0$ and those with the intervention as $\bw_1$. Then, the effect of the intervention can be described as a function of covariates by
$$\tau_{\bw_1,\bw_0}(\bx) \equiv \E\{Y(\bw_1)-Y(\bw_0)|\bX = \bx\},$$
and the extent of treatment effect heterogeneity is dictated by how strongly this depends on $\bX$. If interest is in marginal treatment effects, this quantity can be averaged over the covariate distribution to obtain a marginal effect
$$\overline{\tau}_{\bw_1,\bw_0} = \E[\tau_{\bw_1,\bw_0}(\bX)].$$
This problem effectively reduces to the binary treatment setting where the conditional average treatment effect (CATE) is originally defined as there are only two exposure levels of interest. Because of this, many existing approaches to summarizing effect heterogeneity from the existing literature can be used. The variance of this function, $\text{Var}[\tau_{\bw_1,\bw_0}(\bX)]$, or the treatment effect variable importance measures of \cite{hines2022variable} can be incorporated analogously. For this reason, in the following section, we focus on the more complex setting where we have more than two exposure levels of interest, yet we want to study heterogeneity and examine the impact of each covariate in the degree of heterogeneity.

\subsection{Multivariate Treatment Effect Variable Importance Measures}

In many scenarios, we have more than two exposure levels of interest, yet we still want to define relevant quantities that provide information about treatment effect heterogeneity and which covariates are driving this heterogeneity. For this section, we assume that we have a reference level of the exposure denoted by $\bw_0$, which is chosen a priori. While the following estimands are well-defined for any choice of $\bw_0$ as long as the positivity assumption holds, we recommend selecting $\bw_0$ within a reasonable range, such as the mean of the exposures or a value between their first and third quartiles, to reduce the extent of extrapolation required. In Appendix \ref{sec:Appendix_w0}, we evaluate how robust the results of Section \ref{sec:MedicareAnalysis} are to the choice of $\bw_0$, and find that results are stable across different choices.

With $\bw_0$ fixed, we can define the following quantity:
$$\tau_{\bw_0}(\bx,\bw) \equiv \E[Y(\bw)-Y(\bw_0)|\bX = \bx],$$
which describes how the potential outcome surface varies by both $\bx$ and $\bw$. This function is difficult to interpret as it is a function of two multivariate arguments. We propose multivariate treatment effect variable importance measures (MTE-VIM) to summarize how much the heterogeneity of $\tau_{\bw_0}(\bx, \bw)$ is driven by each particular covariate.

Before defining the variable importance metric, we must first define the overall amount of heterogeneity of the treatment effect, which we define as
$$\phi = \E_{\bW}[\var_{\bX}\{\tau_{\bw_0}(\bX, \bW)\}].$$
This looks at the variability of the treatment effects as a function of $\bX$ for a fixed exposure level, but this variability may differ for different values of exposures $\bW$, and therefore we average this variability across the range of exposures. In practice, the expectation and variance are taken with respect to the empirical distributions of exposures and covariates, respectively. Note that instead of the marginal variance of $\bX$, we could have alternatively used the conditional variance of $\bX$ given $\bW$. We believe, however, that the marginal variance is a more appropriate choice for measuring the amount of treatment effect heterogeneity, as it reflects variation of treatment effects that exist regardless of which units receive each level of exposure. Further, we will see in Section \ref{sec:MTEestimation} that the conditional variance would complicate our estimation of this quantity as it would require a model for the distribution of $\bX$ given $\bW$, which is difficult to specify given the dimension of both the covariates and exposures. Next, we define
$$\phi_j = \E_{\bW}\left(\var_{\bX_{-j}}\left[\E_{X_j|\bX_{-j}}\{\tau_{\bw_0}(\bX, \bW)\}\right]\right)$$
where $\bX_{-j}$ denotes the $p-1$ remaining covariates without the $j$-th covariate $X_j$. This is similar to $\phi$, but we first take the conditional expectation of $\tau_{\bw_0}(\bX, \bW)$ with respect to the $j$-th covariate given the remaining covariates so that the term inside the variance is a function of only $\bX_{-j}$ and $\bW$. We know by the law of total variance that $\phi-\phi_j \in [0,\phi]$, and this quantity measures the amount heterogeneity of $\tau_{\bw_0}(\bX, \bW)$ that cannot be explained without $X_j$.

Assuming some degree of heterogeneity, i.e. $\phi>0$, we propose the following estimand, which we call MTE-VIM, to describe the importance of each covariate to the overall heterogeneity of the causal effect:
$$\psi_j = 1 - \phi_j/\phi \in [0,1].$$ 
This can be interpreted as the proportion of the treatment effect heterogeneity of $\tau_{\bw_0}(\bX, \bW)$ not explained by $\bX_{-j}$. Under the proposed model in Section \ref{sec: estimation}, which implies that the treatment effect heterogeneity is additive in covariates, the sum of all $\psi_j$ equals one when covariates are mutually independent. We note that one should carefully interpret $\psi_j$ in the presence of highly correlated covariates. If two covariates are highly correlated, then they will both have small values of $\psi_j$ regardless of whether they modify the treatment effect. This issue is not unique to our estimand, as it is present in other commonly used variable importance measures that focus on prediction instead of treatment effect heterogeneity \citep{verdinelli2024feature}. We don't view this as a problem, as these estimands are simply descriptive measures aiming to assign relative importance to each covariate in the heterogeneity of the treatment effect. If one is concerned about this issue, then they can used a grouped version of the MTE-VIM, given by $\psi_s$, where $s$ is a subset of $\{1,2,...,p\}$. A detailed description of this can be found in \cite{hines2022variable}, though the general idea is to find the proportion of the treatment effect heterogeneity that can't be explained by the remaining covariates not in $s$. One can use a priori knowledge to select groups of covariates, or use the sample correlation matrix of the covariates to identify groups of correlated variables.

Lastly, while our paper focuses mostly on treatment effect heterogeneity and the role that covariates play in modifying the treatment effect, similar ideas could be used to identify which treatments have the largest impact on the outcome. Specifically, one could define an estimand based on $\E_{\bX}[\var_{\bW}\{\tau_{\bw_0}(\bX, \bW)\}]$ and $\E_{\bX}(\var_{\bW_{-k}}[\E_{W_k\vert \bW_{-k}}\{\tau_{\bw_0}(\bX, \bW)\}])$ to identify the proportion of variability in the treatment effect that can not be explained without treatment $k$, which would provide insights about which treatments are driving the treatment effect.

\section{Estimation Issues}\label{sec: estimation}

Under the assumptions introduced in Section \ref{sec: estimands}, for any fixed $\bw$ we can identify $\E\{Y(\bw)|\bX = \bx\}$ from the observed data by $\E(Y|\bX = \bx,\bW=\bw)$. Other identification strategies such as those involving propensity scores \citep{rosenbaum1983central} or combinations of outcome models and propensity scores \citep{bang2005doubly} are common in causal inference. However, these do not apply here due to the difficulty of estimating the density of the multivariate treatment given the covariates, and because these identification strategies are not well studied for multivariate, continuous treatments. Due to this identification result, all estimands, including the MTE-VIM are identifiable given the conditional outcome regression and therefore we focus our estimation on this quantity. 

We separate the conditional outcome regression surface into three parts: the main effect of covariates $\bX$, the main effect of exposures $\bW$, and the interactions of covariates and exposures on the outcome. Therefore, we write our model as follows:
\begin{align}\label{eq: model}
    \E (Y|\bX,\bW) = c + f(\bX) + g(\bW) + h(\bX,\bW).
\end{align}
Note that this model is overparameterized as it currently stands because the $f(\cdot)$ and $g(\cdot)$ functions can be absorbed into the $h(\cdot, \cdot)$ function. For this reason, these functions are not individually identifiable, and only their sum is identified. We write the model in this way so that we can explicitly shrink each of these components separately, and we discuss in this section how to structure the model so that each of these functions is individually identifiable and therefore amenable to shrinkage and regularization. 

Let $\mathcal{X}_j$ and $\mathcal{X}_{-j}$ denote the support of the $j$-th covariate and the remaining covariates, respectively. To simplify the structure of interactions between covariates and exposures, we make the following assumption:
\begin{assumption} For any $\bw_1,\bw_0\in\mathcal{W}$, $x_j^*, x_j \in \mathcal{X}_j$, and $\bx_{-j}^*, \bx_{-j} \in \mathcal{X}_{-j}$, 
\begin{align*}
    &\E[Y(\bw_1)-Y(\bw_0)|X_j=x_j^*, \bX_{-j}=\bx_{-j}] - \E[Y(\bw_1)-Y(\bw_0)|X_j=x_j, \bX_{-j}=\bx_{-j}] \\
    &= \E[Y(\bw_1)-Y(\bw_0)|X_j=x_j^*, \bX_{-j}=\bx_{-j}^*] - \E[Y(\bw_1)-Y(\bw_0)|X_j=x_j, \bX_{-j}=\bx_{-j}^*]
\end{align*}
for $j=1,...,p$.
\end{assumption}
This assumption ensures that the treatment effect heterogeneity attributable to the $j$-th covariate does not depend on the other covariates. In our model, this assumption implies that $h(\bX,\bW)=\sumjo[p] h_j(X_j,\bW)$, which restricts interactions between covariates and exposures to be limited to interactions between a single covariate and the multivariate exposures. Given that most environmental mixture studies assume no treatment effect heterogeneity, which is a much stronger assumption, and that previous studies for a single exposure have suggested that the treatment effect of \pmtpfs on mortality rate is primarily modified by one covariate at a time \citep{bargagli2020causal, lee2021discovering}, we feel this is a mild and reasonable assumption in the environmental health context. In simulation studies in Appendix \ref{sec:AppendixNonAdditive}, even when this assumption is violated, we find that our proposed model provides a reasonable estimate compared to existing methods. Nonetheless, if strong, higher order covariate interactions with multivariate exposure effects are expected, we recommend using a more complex model, such as a single BART model with both $\bX$ and $\bW$ as input variables that places no additivity restrictions on the model. 

The first hurdle to identification of the individual functions is that $h_j(X_j,\bW)$ may capture the main effect of either the covariate or exposures. For example, we could have $h_j(X_j,\bW) = X_j + W_1 + X_jW_1$, which captures both main effects and interaction terms, when ideally it would only capture $X_jW_1$. Effectively we want this function to only be nonzero if there is truly an interaction between the exposures and $X_j$. For this reason, we restrict $h_j(X_j,\bW)$ to be of the form $h^{cov}_j(X_j)h^{exp}_j(\bW)$ where $h^{cov}_j$ is a non-constant function of covariate $j$, and $h^{exp}_j$ is a non-constant function of the exposures. This allows us to prevent the interaction function $h_j(X_j,\bW)$ from simply absorbing the main effects of $X_j$ and $\bW$. We describe our strategy for each of these separable functions in the following sections.

\subsection{Identification through Shifting}\label{sec: identification}

The formulation for $h(\bX,\bW)$ in the previous section helped ensure that the interaction functions do not solely contain main effects of the exposures or covariates, however, the individual functions are still only identifiable up to constant shifts. For example, without further restriction, shifting $f(\bX)$ upward by $\delta$, and $g(\bW)$ and $h(\bX,\bW)$ downward by $\delta/2$ leads to the same likelihood. This is problematic as we would like our interaction functions to be nonzero only when there is heterogeneity of the treatment effect. If these functions are not identifiable, it becomes more difficult to shrink them towards zero when there is no heterogeneity in the model.

One common way to address this is to put moment restrictions on the functions such that $\E_{\bX}\{f(\bX)\}=\E_{\bW}\{g(\bW)\}=\E_{\bX,\bW}\{h(\bX,\bW)\}=0$ with the additional restriction that $\E_{\bX}\{h(\bX,\bW)\}=\E_{\bW}\{h(\bX,\bW)\}=0$. This approach is difficult to implement, however, as enforcing these conditions requires the conditional density of covariates given exposures, and vice versa, which is difficult to estimate given the dimension of the exposures and covariates. We use an alternative restriction that also leads to identifiability of individual functions, but is straightforward to implement. Specifically, we enforce that $f\{\E(\bX)\}=g\{\E(\bW)\}=0$ and $h\{\E(\bX),\bw\}=h\{\bx,\E(\bW)\}=0$ for all $\bx$ and $\bw$. Under this restriction, the sum of the constant term and $f(\bX)$ represents the conditional expectation of the potential outcome as a function of covariates when exposures are set to their mean. The term $g(\bW)$ represents the shifted average exposure-response curve at the mean of covariates, with the shift ensuring that it equals zero at the mean of exposure. Finally, $h(\bX,\bW)$ captures the remaining treatment effect heterogeneity. We show in Appendix \ref{sec: identifiability proof} that this leads to identifiability of the individual functions. In practice, this restriction is achieved by centering our estimated functions at each iteration of an MCMC algorithm. If we let all functions with a 0-subscript denote unrestricted functions that do not enforce the aforementioned constraints, then these can be written as:
{\small\begin{align*}
    \E (Y|\bX,\bW) &= c_0 + f_0(\bX) + g_0(\bW) + \sumjo[p] h_{j0}(X_j,\bW)\\
    &= \bigg[c_0 + f_0(\E(\bX)) + g_0(\E(\bW)) + \sumjo[p] h_{j0}(\E(X_j),\E(\bW))\bigg]\\
    &\v\v+ \bigg[f_0(\bX) - f_0(\E(\bX)) + \sumjo[p]\left\{ h_{j0}(X_j,\E(\bW)) - h_{j0}(\E(X_j),\E(\bW))\right\}\bigg]\\
    &\v\v+ \bigg[g_0(\bW) - g_0(\E(\bW)) + \sumjo[p]\left\{ h_{j0}(\E(X_j),\bW) - h_{j0}(\E(X_j),\E(\bW))\right\}\bigg]\\
    &\v\v+ \bigg[\sumjo[p] \left\{h_{j0}(X_j,\bW) - h_{j0}(X_j,\E(\bW)) - h_{j0}(\E(X_j),\bW) + h_{j0}(\E(X_j),\E(\bW)) \right\}\bigg]\\
    &=c + f(\bX) + g(\bW) + \sumjo[p] h_j(X_j,\bW).
\end{align*}}
At each iteration of the MCMC, the functions are shifted to satisfy this condition, and therefore our individual functions are identifiable, and the interaction functions should only be nonzero when there is heterogeneity of the treatment effect. Because of this, we can apply shrinkage priors to the $h_j(X_j, \bW)$ functions, which should improve estimation when heterogeneity is not present. Note that the alternative restriction that we impose is baseline-free in the sense that one can replace the chosen baseline, $\E(\bX)$ and $\E(\bW)$, with any values of the covariates and exposures. 

\subsection{Soft Bart and Targeted Smoothing}

In this section we detail the specific models we use for each of the $f(\bX)$, $g(\bW)$, and $h(\bX,\bW)$ functions. For the main effect functions we use the SoftBART prior of \cite{linero2018bayesian}, while we use BART with \textit{targeted smoothing} (tsBART, \cite{starling2020bart, li2022adaptive}) for estimation of the interaction functions. Before detailing each of these BART extensions, we first briefly introduce the original BART model. BART was introduced by \cite{chipman2010bart} and is a fully Bayesian ensemble-of-trees model that has seen increasing usage in causal inference due to the success of the original BART and its variants \citep{dorie2019automated}. Specifically, it assumes that
\begin{align*}
    Z &= f(\boldsymbol{v}) + \epsilon, \qquad \epsilon \overset{\iid}{\sim} N(0,\sigma^2)\\
    f(\boldsymbol{v})&= \sum_{m=1}^M \text{Tree}(\boldsymbol{v}, \mathcal{T}_m, \mathcal{M}_m)
\end{align*}
where $\text{Tree}(\boldsymbol{v}, \mathcal{T}_j, \mathcal{M}_j)$ represents a regression tree with tree structure $\mathcal{T}_j$ inducing a step function in $\boldsymbol{v}$, and leaf parameters $\mathcal{M}_m=\{\mu_{m1},\cdots, \mu_{mb_m}\}$ for prediction. The standard prior distribution sets $\mu_{mk} \sim N(0,\sigma_\mu^2/M)$ so that the overall prior variance is $\var(f(\bv))=\sigma_\mu^2$. Although we do not detail the hyper-parameters of BART in this paper, we note that the default setting of \cite{chipman2010bart} encourages each tree to be shallow and shrinks leaf parameters toward zero so that $\text{Tree}(\cdot)$ can be considered as a weak learner. We refer interested readers to \cite{linero2017review} and \cite{hill2020bayesian} for recent reviews of BART.

Although the original BART approach is successful, due to the piece-wise constant nature of tree ensembles, it suffers when the underlying truth is smooth even if it is relatively simple. To overcome this shortcoming, \cite{linero2018bayesian} propose a smooth modification of BART (SoftBART or SBART) by allowing $\boldsymbol{v}$ to follow a probabilistic path, that is, randomizing the decision rule at each split. Because we expect the effects of environmental exposures or other continuous treatments on the outcome to be smooth, we use SoftBART for estimation of both $f(\bX)$ and $g(\bW)$. 

Another smooth variant of BART is tsBART \citep{starling2020bart, li2022adaptive} which is originally motivated by time-to-event data and density estimation, though we use it for estimation of the interaction functions $h_j(X_j, \bW)$. Unlike SoftBART which smooths the regression function over all variables, tsBart smooths over a single targeted variable, say $u$. After centering at $\gamma(u)$, a baseline function of $u$, the model can be written as $f(u,\boldsymbol{v})= \gamma(u) + \sumjo[m] \text{tsTree}(u,\boldsymbol{v}, \mathcal{T}_m, \mathcal{M}_m)$ where each terminal node is associated with $\mathcal{M}_m=\{\mu_{m1}(u),\cdots, \mu_{mb_m}(u)\}$. Here, each element $\mu_{mb}(u)$ follows a Gaussian process in $u$ with mean zero and covariance function $\Sigma(u,u')$, so that  $f(u,\boldsymbol{v})$, the sum of $m$ Gaussian processes, is continuous in $u$. Following \cite{li2022adaptive}, we approximate the model by $f(u,\boldsymbol{v})= \gamma(u) + \sum_{m=1}^M \mathcal{B}_m(u)\text{Tree}(\boldsymbol{v}, \mathcal{T}_m, \mathcal{M}_m)$ where $\mathcal{B}_m(u) = \sqrt{2} \cos(\omega_m u + b_m)$ is a random basis function where $\omega_m \iidsim N(0, \rho^{-2})$ and $b_m \iidsim U(0,2\pi)$ so that $\sum_{m=1}^M \mathcal{B}_m(u)\text{tree}(\boldsymbol{v}, \mathcal{T}_m, \mathcal{M}_m)$ weakly converges to a Gaussian process $\text{GP}(0,\Sigma(\cdot,\cdot))$ as $M\rightarrow\infty$ where $\Sigma(u,u') = \sigma_\mu^2 \exp\{-(u-u')/(2\rho^2)\}$.

By fitting tsBART for each interaction function $h_j(X_j,\bW)$ in our model by setting $u=X_j$ and $\bv = \bW$, we have that each fitted function is necessarily the sum of products of a continuous function of $X_j$ and a flexible tree of $\bW$ so that it does not include a function solely of $X_j$ or $\bW$. The default prior specifications from \cite{linero2018bayesian} and \cite{li2022adaptive} are used throughout our implementation. However, to enhance stability, we regularize the interaction function more aggressively by capping the variance of each terminal node, $\sigma_\mu$, at $3.5/2\sqrt{M}$, which is its default initial specification.

\subsection{Estimation of MTE-VIM}

\label{sec:MTEestimation}
We now detail the estimation strategy for the proposed variable importance measures. Recall that MTE-VIM for the $j$-th covariate, $\psi_j$, is composed of the total heterogeneity $\phi = \E_{\bW}[\var_{\bX}\{\tau_{\bw_0}(\bX, \bW)\}]$ and the total heterogeneity accounted for by $\bX_{-j}$, given by $\phi_j= \E_{\bW}\left(\var_{\bX_{-j}}\left[\E_{X_j|\bX_{-j}}\{\tau_{\bw_0}(\bX, \bW)\}\right]\right)$. As before, under the assumptions of Section \ref{sec: estimands}, we can write $\tau_{\bw_0}(\bx, \bw) = \E(Y|\bX=\bx,\bW=\bw) - \E(Y|\bX=\bx,\bW=\bw_0)$, and therefore the posterior distribution of $\tau_{\bw_0}(\bx, \bw)$ can be obtained once we have the posterior distribution of the outcome regression model. Obtaining $\phi$ is relatively straightforward, as we use empirical distributions of both the exposures and covariates to approximate moments with respect to $\bW$ or $\bX$. We can construct an $n\times n$ matrix for which the $(k,l)$-th element is $\tau_{\bw_0}(\bX_{l}, \bW_{k})$, where $\bX_{l}$ and $\bW_{k}$ represent the observed covariates of the $l$-th observation and the exposures of the $k$-th observation, respectively. Then, the sample variance of the $k$-th row of the matrix would estimate $\var_{\bX}\{\tau_{\bw_0}(\bX, \bW_{k})\}$. Taking the sample mean of these $n$ column sample variances estimates $\E_{\bW}[\var_{\bX}\{\tau_{\bw_0}(\bX, \bW)\}]$. This estimation process can be illustrated as follows:
{
\begin{align*}
\begin{pmatrix}
\tau_{\bw_0}(\bX_{1}, \bW_{1}) & \tau_{\bw_0}(\bX_{2}, \bW_{1}) & \dots & \tau_{\bw_0}(\bX_{n}, \bW_{1}) \\
\tau_{\bw_0}(\bX_{1}, \bW_{2}) & \ddots &  & \tau_{\bw_0}(\bX_{n}, \bW_{2}) \\
\vdots &  & \ddots &  \vdots \\
\tau_{\bw_0}(\bX_{1}, \bW_{n}) & \tau_{\bw_0}(\bX_{2}, \bW_{n}) & \dots & \tau_{\bw_0}(\bX_{n}, \bW_{n})
\end{pmatrix}
\longrightarrow &
\begin{pmatrix}
\text{Var}_{\bX} [\tau_{\bw_0}(\bX, \bW_{1})] \\[5pt]
\text{Var}_{\bX} [\tau_{\bw_0}(\bX, \bW_{2})] \\ \vdots \\
\text{Var}_{\bX} [\tau_{\bw_0}(\bX, \bW_{n})] \\
\end{pmatrix}\\
& \hspace{0.5 in} \big\downarrow \\
& \E_{\bW}[\var_{\bX}\{\tau_{\bw_0}(\bX, \bW)\}]
\end{align*}}
The same strategy can be used for $\phi_j$ replacing $\tau_{\bw_0}(\bX_{l}, \bW_{k})$ with $\E_{X_j|\bX_{-j}}\{\tau_{\bw_0}(\bX_{l}, \bW_{k})\}$ everywhere, so all that is left is to describe how to obtain $\E_{X_j|\bX_{-j}}\{\tau_{\bw_0}(\bX_{l}, \bW_{k})\}$. For ease of exposition, we re-write $\tau_{\bw_0}(\bX_{l}, \bW_{k})$ as $\tau_{\bw_0}(X_{lj}, \bX_{l(-j)}, \bW_{k})$ where $X_{lj}$ denotes the $j$-th covariate of the $l$-th individual and $\bX_{l(-j)}$ denotes the remaining covariates of the $l$-th individual. A nonparametric estimate of this mean can be defined as
$$\widehat{\E}_{X_j|\bX_{-j}}\{\tau_{\bw_0}(\bX, \bW)\} = \frac{\sum_{i=1}^n \tau_{\bw_0}(X_{ij}, \bX_{-j}, \bW) K(\bX_{i(-j)} - \bX_{-j})}{\sum_{i=1}^n K(\bX_{i(-j)} - \bX_{-j})}$$
where $K(\cdot)$ is an appropriately defined kernel. We do not provide specific details regarding the choice of the kernel or the bandwidth parameter, as standard approaches to selection of these parameters in nonparametric regression could be applied. The sample correlation matrix of the covariates could help inform this decision, as a larger bandwidth would be appropriate if $X_j$ is assumed to be (approximately) independent of the other covariates. In this case, kernel smoothing simplifies to the sample average of $\tau_{\bw_0}(X_{ij}, \bX_{-j}, \bW)$, which we do in the simulation studies in Section \ref{sec: simulation} where covariates are mutually independent. For a comprehensive overview of kernel smoothing techniques, see \cite{hastie2009kernel, ghosh2018kernel}. While the kernel approach is applicable if the dimension of $\bX$ is small, it won't work well as the number of covariates grows. In this more difficult setting, we could make parametric assumptions and estimate the mean of $X_j$ given $\bX_{-j}$ using a regression model. For instance, if $X_j$ is continuous, we could assume that it follows a normal distribution and we can estimate the conditional mean and variance using linear regression. This would provide us an estimate of the conditional density $f_{X_j | \bX_{-j}}$ from which we can calculate $\E_{X_j|\bX_{-j}}\{\tau_{\bw_0}(\bX, \bW)\}$. We use this regression-based approach in the data analysis in Section \ref{sec:MedicareAnalysis}. Since the kernel smoothing and parametric regression approaches discussed above pertain specifically to the calculation of the MTE-VIM, which is a purely descriptive measure of the treatment effect function, they are separated from the MCMC sampling. Also, we believe that the MTE-VIM, which aims to provide some degree of interpretability of the multidimensional treatment effect function, remains useful even under mild misspecification of the covariate distribution. In principle, one could employ flexible Bayesian nonparametric models for estimating the covariate distribution and incorporate this procedure into the MCMC sampling algorithm to improve the accuracy of MTE-VIM estimation. We do not pursue that approach here as it would significantly increase computation time, especially for large datasets such as those used in our analysis.

Note that both of these require the construction of an $n\times n$ matrix at each MCMC sample, which is computationally intensive. This is particularly problematic for the analysis of Medicare data in Section \ref{sec:MedicareAnalysis}, where our sample consists of all 38702 zip-codes in the United States. As an approximation to this calculation, one can use a blocking scheme where we construct several sub-matrices of the data when calculating the variable importance metrics. Specifically, we split the overall sample into $K$ groups of approximately equal size. Letting the sample size in group $k$ be $n_k$, we can estimate the variable importance metric using the subset of the data in group $k$, which only requires the calculation of an $n_k \times n_k$ matrix. We can do this for each group separately, and then average results across groups for a final estimate of $\psi_j$ for $j =1, \dots, p$. We see in Appendix \ref{sec:AppendixBlocking} that this blocking scheme performs equally well as using the full sample, but is significantly faster on large data sets. Therefore, we use our method with the blocking scheme throughout the paper. We also note that the blocking scheme is entirely independent of the fitting of the conditional outcome models, and therefore, it should not affect the near-minimax contraction rates for estimating conditional outcomes, which we provide in Section \ref{sec: contraction rates}.

\section{Posterior contraction rates}\label{sec: contraction rates}
Next, we study the rate at which the posterior distributions concentrate in the model described in Section \ref{sec: estimation}. Let $\mathcal{S}(\alpha, p, d)$ denote a set of $\alpha-$H\"older continuous functions on $[0,1]^p$ that are constant in all but $d$ of the coordinates. Following \cite{linero2018bayesian} and \cite{li2022adaptive}, we make the following assumptions about the data-generating process (Condition A).
\begin{enumerate}[label= (A\arabic*)]
    \item $Y_i\sim\nnormal(c^*+f^*(\bX_i)+g^*(\bW_i)+h^*(\bX_i,\bW_i), 1)$
    \item The range of $\bX_i$ and $\bW_i$ are $[0,1]^p$ and $[0,1]^q$, respectively.
    \item $f^*\in \mathcal{S}(\alpha_x, p, d_x)$, $g^*\in \mathcal{S}(\alpha_w, p, d_w)$, and $h^*\in \mathcal{S}(\alpha_h, p+q, d_h)$
\end{enumerate}
We assume the error variance is $\sigma^2 = 1$ for simplicity, but it is straightforward to incorporate unknown $\sigma^2$ as well. Achieving (A2) is straightforward by applying quantile normalization, where $X_{ij}$, the $i$-th individual's $j$-th covariate, is replaced with its quantile value among all the $j$-th covariate values, and the same for the exposures. Further, we must make assumptions about the SoftBART prior distribution $\Pi$ used for each regression function $f$, $g$, and $h$. For brevity, we leave these specific details to Appendix \ref{sec: proof}, and we refer to these assumptions on the prior distribution as Condition P. We have simplified the setup by not incorporating targeted smoothing into $h(\bx, \bw)$, though we emphasize that \cite{li2022adaptive} show that for certain choices of basis function, it is not difficult to prove analogous results for targeted smoothing models. We remove targeted smoothing from consideration so that we can use the same Condition P for all model components rather than requiring a separate set of conditions for $h(\cdot,\cdot)$.

Let $\E_0$ denote the expectation with respect to the true data-generating mechanism, $\Pi_n$ denote the posterior distribution, $\Vert\mu\Vert_n^2=\dfrac{1}{n}\sumio \mu(\bX_i,\bW_i)^2$ where $\mu(\bX_i,\bW_i)=c+f(\bX_i)+g(\bW_i)+h(\bX_i,\bW_i)$, and $\mu^*$ denote the ground truth of $\mu$. In Appendix \ref{sec: proof}, we prove the following theorem.
\begin{theorem}\label{Thm: Convergence} If Conditions A and P hold, 
    $$\E_0\{\Pi_n(\Vert \mu-\mu^*\Vert_n > M\epsilon_n)\}\rightarrow 0, \mathtext{as} n\rightarrow\infty$$
    for $\epsilon_n=\max(\epsilon_{nf}, \epsilon_{ng}, \epsilon_{nh})$ where
        $\epsilon_{nf} = n^{-\alpha_x / (2\alpha_x + d_x)} \log(n)^{t_f}, 
  \epsilon_{ng} = n^{-\alpha_w / (2\alpha_w + d_w)} \log(n)^{t_g}$ and
  $\epsilon_{nh} = n^{-\alpha_h / (2\alpha_h + d_h)} \log(n)^{t_h}$ where
        $t_f = \alpha_x (d_x + 1) / (2\alpha_x + d_x), 
        t_g = \alpha_w (d_w + 1) / (2\alpha_w + d_w),$ and $t_h = \alpha_h (d_h + 1) / (2\alpha_h + d_h).$
\end{theorem}
\noindent The rates $\epsilon_{nf}, \epsilon_{ng}, \epsilon_{nh}$ represent the minimax estimation rates within the respective function classes for $(f, g, h)$, up to a logarithmic term, and therefore the rate $\epsilon_n$ is a near-minimax optimal rate for estimation of $\mu^*$. This result shows that our model is able to capture any true outcome regression function as long as it is sufficiently smooth, which is reasonable for the study of environmental pollutants. Further note that while we have derived posterior contraction rates for $\mu$, these imply the same posterior contraction rates for causal effects, such as $\tau_{\bw_0}(\bx,\bw)$, and for variable importance metrics $\psi_j$.

\section{Simulation Studies}\label{sec: simulation}

Here, we assess the performance of our proposed approach and compare it to existing approaches that do not allow for heterogeneity of the multivariate treatment effect. For the simulation studies we draw covariates and exposures from the following multivariate normal distributions:
\begin{align*}
    &\bX_i \overset{\iid}{\sim} N(\bo,\bi_5),
    \quad \bW_i \sim N(\Vec{\mu}_i,\Vec{\Sigma})\qquad \text{where}\\
    &\Vec{\mu}_i = 
    \left(e^{X_{i1}}/(1+e^{X_{i1}}) - 0.5, 0.1X_{i2}^2 - 0.1, 0.3 X_{i3}, \sin (X_{i2}), 0.05 X_{i4}^3\right)^T,\\
    &\Sigma_{ij}= \begin{cases}
    1\qquad &i=j\\
    0.3\qquad &i\neq j
    \end{cases}
\end{align*}
and $\Sigma_{ij}$ denotes the $(i,j)$ element of $\Vec{\Sigma}$. The covariates are associated with the exposures and the exposures are correlated with each other, which are both common in observational studies in environmental health. Next, we generate the outcome by
\begin{align*}
    &Y_i = f^*(\bX_i) + g^*(\bW_i) + \sumjo[5] h^*_j(X_{ij},\bW_i) + \epsilon_i, \qquad \epsilon_i \iidsim N(0,1)\\
    &f^*(\bX_i) = X_{i1} + X_{i2} - 0.5X_{i3}\\
    &g^*(\bW_i) = \indct{W_{i1}>0} + W_{i1}e^{0.3W_{i3}} + \arctan(W_{i2}) + \sin(W_{i2}W_{i3}\pi)+\min(|W_{i3}|,1).
\end{align*}
Further, we consider three scenarios of interactions between covariates and exposures:
{\small\begin{align*}
    \textbf{No interaction: } &h^*_1(X_{i1},\bW)=h^*_2(X_{i2},\bW)=0\\
    \textbf{Moderate: } &h^*_1(X_{i1},\bW) = 0.2\arctan (4X_{i1})g^*(\bW),\quad h^*_2(X_{i2},\bW) = 0.2\cos (X_{i2}\pi)g^*(\bW)\\
    \textbf{Strong: } &h^*_1(X_{i1},\bW) = 0.4\arctan (4X_{i1})g^*(\bW),\quad h^*_2(X_{i2},\bW) = 0.4\cos (X_{i2}\pi)g^*(\bW)
\end{align*}}
while $h^*_j(X_{ij},\bW)=0$ for $j=3,4,5$ under all scenarios. The standard deviations of the heterogeneous treatment effect functions are approximately 0.2 and 0.4 times that of the sum of the two main effects, respectively, which is reasonable as we expect that interaction effects to be smaller in magnitude than main effects in general. We run the simulation for 300 different data replicates for each scenario and average results across all simulated data sets. We set the sample size to be $n=2000$. We aim to estimate $\tau_{\bw_1,\bw_0}(\bX) \equiv \E\{Y(\bw_1)-Y(\bw_0)|\bX\}$ at 100 randomly chosen locations from the distribution of $\bX$, as well as the variable importance metrics $\psi_j$ for all covariates. We choose two fixed exposures levels $\bw_0=(-0.5, -0.5,-0.5,-0.5,-0.5)^T$ and $\bw_1=(0.5, 0.5,0.5,0.5,0.5)^T$ which roughly correspond to the first and the third quartiles of each exposure, respectively. We consider four different approaches in the simulation. The first is the Bayesian kernel machine regression approach (BKMR) commonly used in the environmental mixtures literature, which fits a model of the form $E(Y \vert \bX, \bW) = \bX \boldsymbol{\beta} + m(\bW)$. A Gaussian process prior is placed on $m(\cdot)$, which incorporates spike-and-slab prior distributions to remove unnecessary exposures. The second is a standard BART prior that fits a model of the form $E(Y \vert \bX, \bW) = m(\bX, \bW)$ and places a BART prior on $m(\cdot)$ (BART). The third uses a smooth variant of BART for $m(\cdot)$ (SoftBART). Lastly, we use our proposed approach with the blocking scheme (SepBART), which separates the regression surface into separate functions and then uses either SoftBART or tsBART prior distributions for each function separately. Note that BKMR specifically assumes that the effect of covariates is linear, which is the case in our simulation studies, while it does not allow for treatment effect heterogeneity, which makes it misspecified in scenarios with interactions between exposures and covariates.

Figure \ref{fig:CATE} shows the simulation results for estimating $\tau_{\bw_1,\bw_0}(\bX)$. We find that our method is the only method that achieves the nominal 95\% coverage for all interaction scenarios and produces the smallest root mean squared error (RMSE) when there is treatment effect heterogeneity, regardless of its strength. While it is expected that BKMR does not perform well in the presence of interactions between covariates and exposures as it assumes no treatment effect heterogeneity, it is surprising that SepBART outperforms BKMR even when there is no heterogeneity and the main effect of covariates is linear, which is the situation that BKMR is designed for. The no-interaction scenario is the least favorable setting for SepBART as it forces an interaction structure on the model. However, even for this setting, SepBART shows comparable estimation performance to the best performing model in terms of RMSE, which indicates that our model is able to shrink the $h(\cdot)$ functions to zero when they are not required. It is also notable that SepBART improves on the original BART and SoftBART that do not separate the regression function into distinct components, which suggests that we benefit from the proposed model formulation and separation of the effects into different functions by providing balance between flexibility and efficiency.
\begin{figure}
    \centering
    \includegraphics[width=5 in]{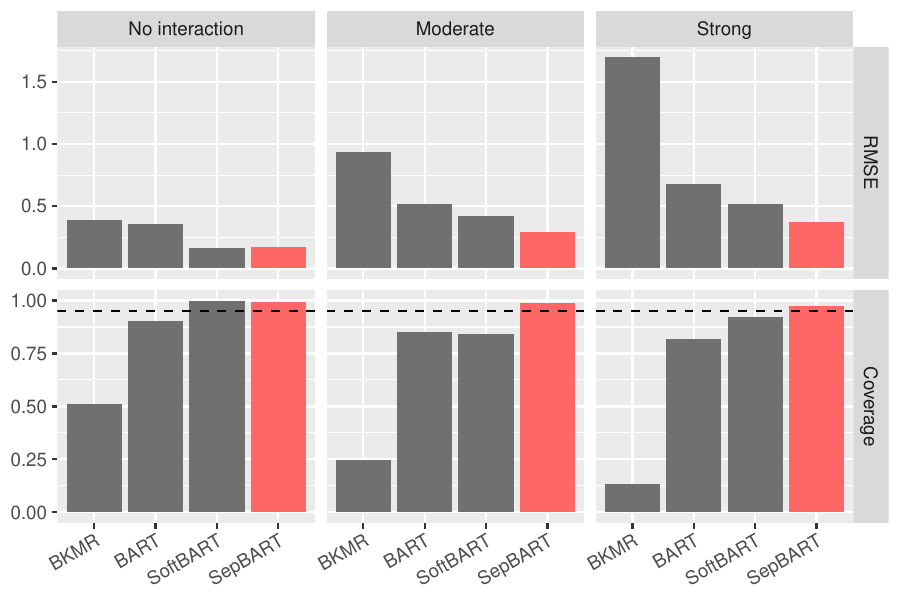}
    \caption{Simulation results for estimating CATE when $\bw_0$ and $\bw_1$ are given. The dashed line in the second row denotes nominal 95\% coverage.}
    \label{fig:CATE}
\end{figure}

We now assess the performance for estimating MTE-VIM introduced in Section \ref{sec: estimands}. As BKMR does not allow for effect heterogeneity, we focus on our approach only here. We calculate $\hat{\psi}_j$ for $j=1,2,...,5$ for each data set by taking the posterior mean of estimates and plot the distribution of $\hat{\psi}_j$ in Figure \ref{fig:VIM}. The true values of $\psi_1$ and $\psi_2$ are 0.72 and 0.28, respectively, and the others are zero. We see that for both interaction scenarios where the true total heterogeneity $\phi$ is 0.26 and 1.14, respectively (moderate and strong), the posterior mean of each estimate is concentrated near the true value. Along with estimating MTE-VIM, we can conduct a hypothesis test for the difference between two MTE-VIMs, i.e. testing $H_0: \psi_j=\psi_k$. This can be done by constructing the posterior distribution of $\psi_j-\psi_k$ and seeing whether or not the interval contains zero. We calculate how often the null hypothesis is rejected with level $\alpha=0.05$ over 300 data set replicates, and plot the empirical rejection rate of each testing pair in Figure \ref{fig:Test}. Since the true $\psi_1$ is far away from zero, it shows high rejection rates for comparing $\psi_1$ and the others (the first column of Figure \ref{fig:Test}). For $\psi_2$, which is non-zero, but not large as $\psi_1$, it produces slightly weaker power when interactions are moderate, but detects almost all differences when interactions are strong. For the last three columns where the null is true, we see rejection rates are controlled under the desired level $\alpha=0.05$ regardless of the scenario, though inference is somewhat conservative.

\begin{figure}
    \centering
    \includegraphics[width = 4.7 in]{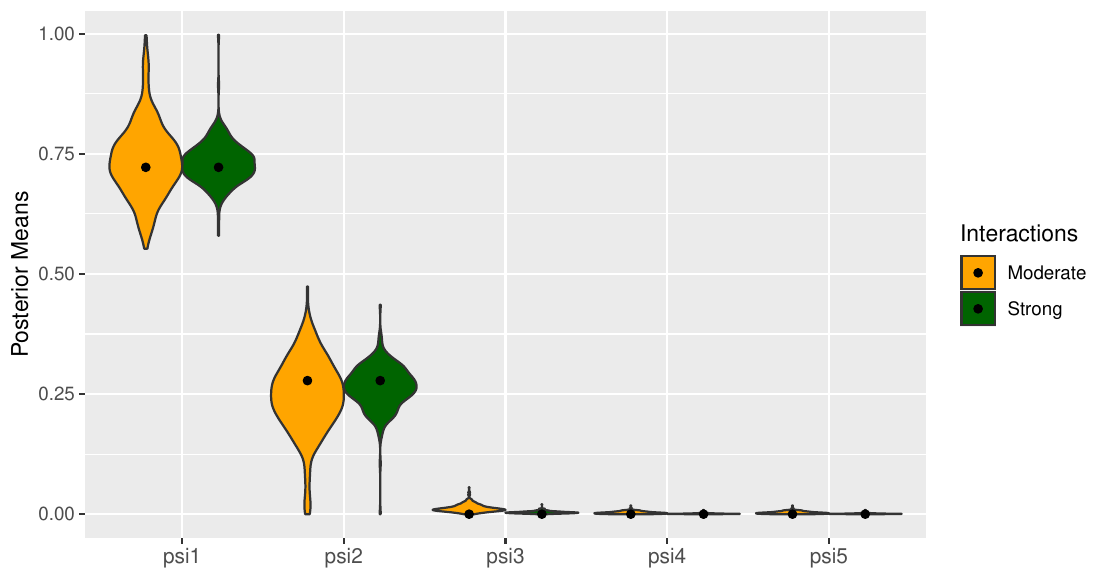}
    \caption{Violin plots of 300 posterior means of the variable importance measures, $\psi_j$ for $j=1,2,...,5$ under each of the two scenarios with heterogeneity. The black dots represent the true variable importance measures: 0.72 for $\psi_1$, 0.28 for $\psi_2$, and 0 for the others.}
    \label{fig:VIM}
\end{figure}

\begin{figure}
    \centering
    \includegraphics[width=5 in]{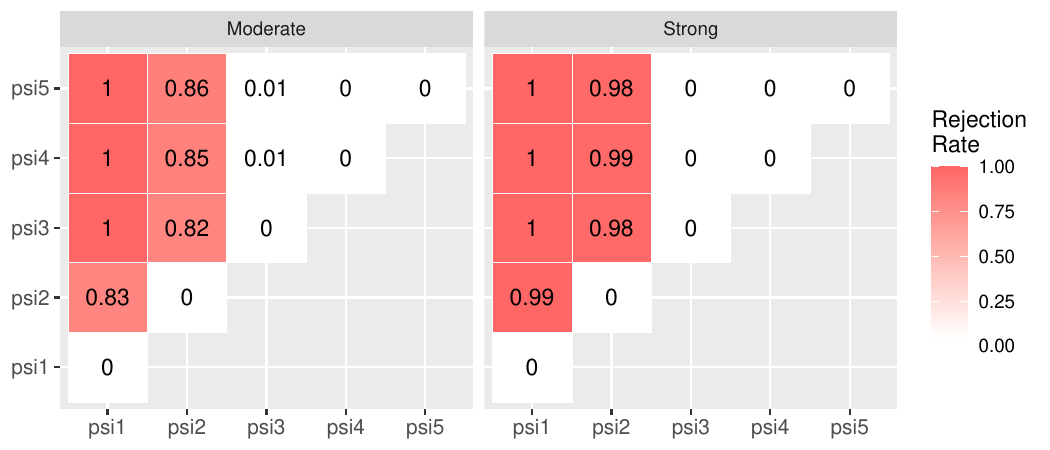}
    \caption{Rejection rates for each test where $H_0: \psi_i=\psi_j$ with the level $\alpha=0.05$.}
    \label{fig:Test}
\end{figure}

\section{Health effects of air pollution mixtures}
\label{sec:MedicareAnalysis}
We now use our proposed approach to gain important epidemiological insights about the effect of PM$_{2.5}$ components and ozone on mortality. We collect information on all Medicare beneficiaries in the United States above the age of 65 for the years 2000-2016. We observe a number of individual-level covariates for the Medicare beneficiaries, such as their age, race, sex, and whether they have dual eligibility to Medicaid, which is a proxy for low socioeconomic status. We also observe a number of area-level covariates unique to each zip code, sourced from the United States Census Bureau and the CDC's Behavioral Risk Factor Surveillance System, regardless of individuals' age. These area-level covariates consist of average body mass index, smoking rates, median household income, education, population density, and percent owner occupied housing. We also adjust for both temperature and humidity variables that are available from the National Climatic Data Center. 

We obtain zip-code level environmental exposure data from two distinct sources. We obtain estimates of total PM$_{2.5}$, ammonium, nitrates, and sulfate levels on a ($0.01^\circ \times 0.01^\circ$) monthly grid from the Atmospheric Composition Analysis Group \citep{van2019regional}. We obtain estimates of elemental carbon, organic carbon, and ozone on a 1km by 1km daily grid from the Socioeconomic Data and Applications Center \citep{di2019ensemble, di2021daily, requia2021daily}. We do not have exact residential addresses of individuals in Medicare and only know their residential zip-code, and therefore all exposures are aggregated to the yearly level at each zip-code. All individual and geographic-level covariates are also aggregated up to the zip-code level by taking their averages or proportions within each zip-code so that the outcome of interest is the annual mortality rate in each zip-code, which we define as 100 times the number of deaths in that zip-code for a particular year divided by the number of person-years in that zip-code.

We run our proposed approach as described in Section \ref{sec: estimation} targeting both marginal and heterogeneous treatment effects to evaluate the extent to which ambient air pollution affects mortality, and whether this effect varies by characteristics of zip-codes. We run our model for each year separately using the prior year exposures as the pollutants of interest, and therefore will present results for all years between 2000 and 2016. We run our MCMC algorithm for 10,000 iterations, discarding the first 4,000 as a burn-in, and thinning every twelfth sample. A detailed discussion on the convergence of MCMC sampling can be found in Appendix \ref{sec:Appendix_convergence}. Overall, we find that convergence diagnostics are very good for the average treatment effect across all years studied. Convergence is slightly worse for MTE-VIM values, though this is likely caused by multi-modality of the posterior distribution, rather than an issue of MCMC sampling, and is still within an acceptable range.

\subsection{Marginal effects of exposures}\label{sec: Data_marginal}

\begin{figure}
    \centering
    \includegraphics[width=5.5 in]{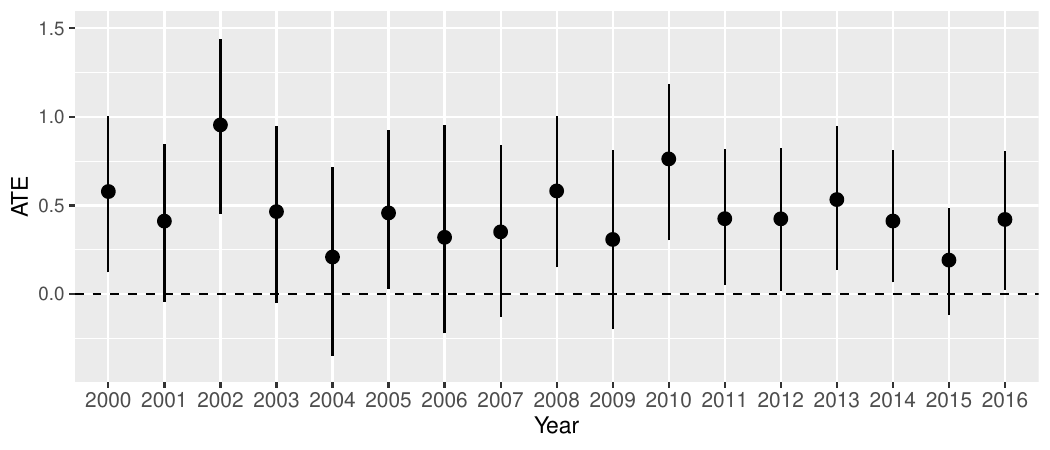}
    \caption{Posterior means and corresponding 95\% credible intervals for the average treatment effects on mortality rates when every pollutant simultaneously increases from its yearly first quartile to the third quartile. The dashed line represents no average treatment effect.}
    \label{fig:data_ATE}
\end{figure}

First, we investigate the average treatment effect of an increase in environmental mixtures, denoted as $\E\{Y(\bw_1)-Y(\bw_0)\}$, where $\bw_0$ and $\bw_1$ represent vectors of all exposures corresponding to the first and third quartiles for each exposure, respectively. Hence, a positive ATE would indicate a harmful effect on mortality due to an increase in environmental exposures. Figure \ref{fig:data_ATE} illustrates estimates of the ATE for each year. We consistently estimate a positive ATE, implying that increasing levels of all environmental exposures increases the mortality rate across the contiguous United States, which suggests a detrimental impact of environmental pollutants on human health.

As previously mentioned, the positivity assumption is a critical consideration when analyzing multivariate exposures, as highlighted in \cite{antonelli2024causal}. To address this, we assess the positivity assumption and examine a trimmed estimator that is more robust to positivity violations, with detailed analyses presented in Appendix \ref{sec:Appendix_positivity}. We evaluated the positivity assumption for each exposure and found that the majority of observations satisfy a univariate positivity criterion for each exposure simultaneously. Additionally, we examined a trimmed ATE, which targets the ATE for observations with the highest probability of being exposed to either $\bw_0$ and $\bw_1$. The results are very similar between the trimmed ATE and the original ATE, which suggests that our choice of $\bw_1$ and $\bw_0$ does not strongly violate the positivity assumption.

\subsection{Heterogeneity of the effect of PM$_{2.5}$ components}\label{sec:data_hetero}

While the average treatment effect provides insights into the health effects of air pollution, of even more interest is whether this effect varies across the population, as it is crucially important to understand which communities are at most risk to the detrimental effects of air pollution. Toward this goal, we first investigate the proposed MTE-VIM for each covariate to understand which characteristics drive the heterogeneity of the causal effect. The estimated total heterogeneity, $\phi$, averaged over all of the study years is 0.016, which implies that the standard deviation of $\tau_{\bw_0}(\bW,\bX)$ with respect to $\bX$ given $\bW=\bw$ is, on average, 0.13. Considering that we estimated an average ATE of 0.46 across study years in Figure \ref{fig:data_ATE}, this suggests there is a non-negligible amount of treatment effect heterogeneity. Figure \ref{fig:data_VIM_by_year} shows the posterior mean of the MTE-VIM corresponding to each covariate for each study year. 

We observe some degree of variation in the MTE-VIM for certain covariates over the study years. This is likely due to the inherent difficulty in estimating these parameters in situations with highly correlated exposures and covariates. For this reason, there will be more variability in the estimates of treatment effect heterogeneity and the corresponding MTE-VIMs compared with the marginal effect estimates seen in Section \ref{sec: Data_marginal}, which leads to more variable results across years. Despite this variability, there is a clear trend that both race and dual eligibility to Medicaid modify the treatment effect the most, followed by age. Figure \ref{fig:data_VIM_avg} shows the value of yearly variable importance metrics in Figure \ref{fig:data_VIM_by_year} when averaged over all study years, and we see that race, dual eligibility to Medicaid, and age achieved the largest variable importance measures, with overall means of 0.24, 0.22, and 0.1, respectively. This indicates that, on average, 24\% of the treatment effect heterogeneity can only be explained by race, even after controlling for other potential effect modifiers. In situations such as this with high variability around estimates of the importance of each covariate on treatment effect heterogeneity, we can also examine grouped MTE-VIMs. In Appendix \ref{sec:Appendix_GVIM}, we placed covariates into one of three groups and found consistent results across years showing that socioeconomic factors (which include race and dual eligibility to Medicaid) had the largest grouped MTE-VIM values.

In addition to identifying the characteristics of zip-codes that modify the treatment effect the most, it is important to understand the nature and direction of this heterogeneity to better understand which groups are most susceptible to air pollution mixtures. To do this, we examine the conditional average treatment effect, using the same reference exposure levels as in the previous section, denoted by $\bw_1$ and $\bw_0$. Specifically, we plot $\E\{Y(\bw_1) - Y(\bw_0) | \bX = \widetilde\bx_{(j)}\}$ as a function of $x_j$, where $\widetilde\bx_{(j)} = (\widebar{x}_{1}, ..., x_j, ..., \widebar{x}_p)$. This sets all covariates to their sample mean and only varies covariate $j$, so that we can examine whether the conditional average treatment effect increases or decreases with covariate $j$. Within the framework of our model and our identification restrictions, $\E\{Y(\bw_1) - Y(\bw_0) | \bX = \widetilde\bx_{(j)}\}$ simplifies to the difference in the interaction function, $h_j(\bw_1, x_j) - h_j(\bw_0, x_j)$, which is zero when $x_j$ equals the mean covariate level, $\widebar{x}_{j}$. Therefore, when this difference is positive, it implies that individuals with the given covariate level are more susceptible to increases in the environmental mixture compared to those with an average covariate level.

Figure \ref{fig:data_hetero} shows the conditional average treatment effect curves described above for the covariates with large MTE-VIMs. Across most years, we observed that the treatment effects of air pollution tend to decrease as the proportion of White populations in a zip code increases, while harmful effects are more pronounced in areas with smaller White populations. This trend suggests a potentially greater impact of air pollution on racial minorities. We found that the treatment effect increases significantly with higher rates of dual eligibility, indicating that the detrimental effects of air pollution are more severe in areas with lower socioeconomic status. We also find that areas with older individuals are more adversely affected by increases in pollution, though this effect is not as strong as that of dual eligibility to Medicaid. Overall, our study shows that results from the previous literature \citep{simoni2015adverse, bargagli2020causal} that focused solely on univariate PM$_{2.5}$ exposure extend to the more complex setting of multivariate air pollution mixtures. We find a harmful effect of air pollution on mortality, and that this effect is heterogeneous with larger effects in areas with low socioeconomic status, lower proportions of white individuals, and older individuals.

\begin{figure}
    \centering
    \includegraphics[width=5 in]{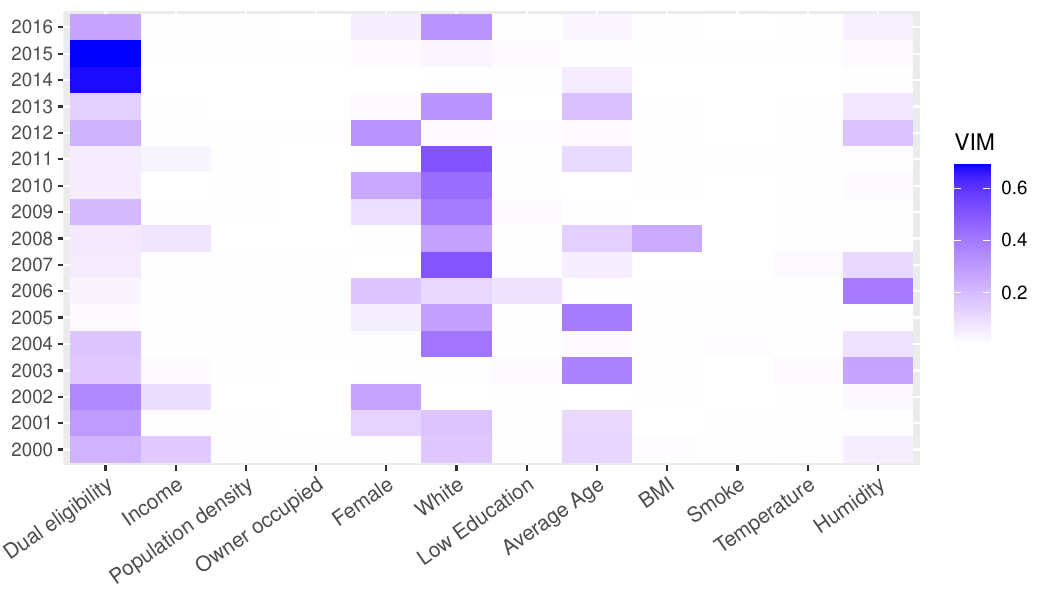}
    \caption{Heatmap of posterior means of the MTE-VIM by year.}
    \label{fig:data_VIM_by_year}
\end{figure}
\begin{figure}
    \centering
    \includegraphics[width=5 in]{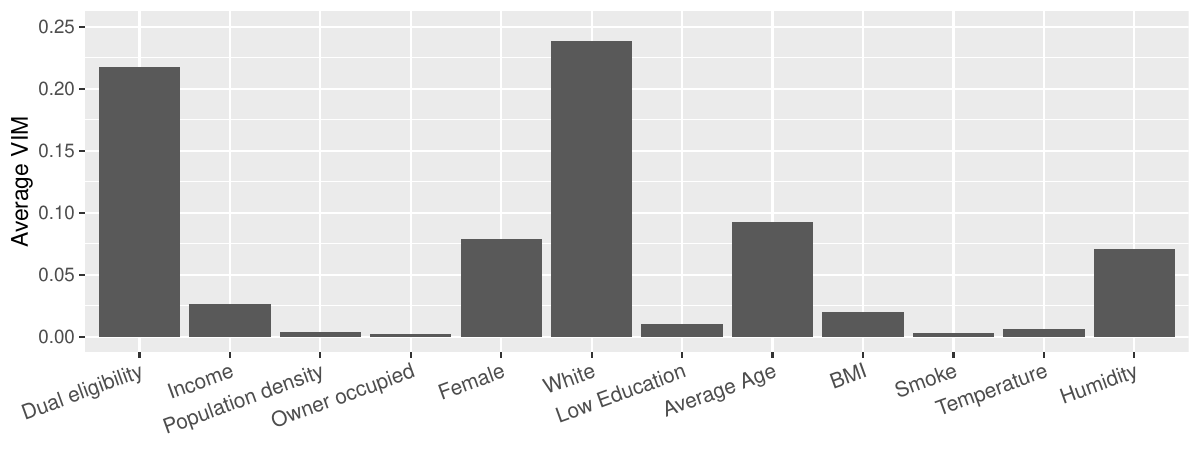}
    \caption{Estimates of the MTE-VIM averaged over the years 2000-2016.}
    \label{fig:data_VIM_avg}
\end{figure}
\begin{figure}
    \centering
    \includegraphics[width=5.5 in]{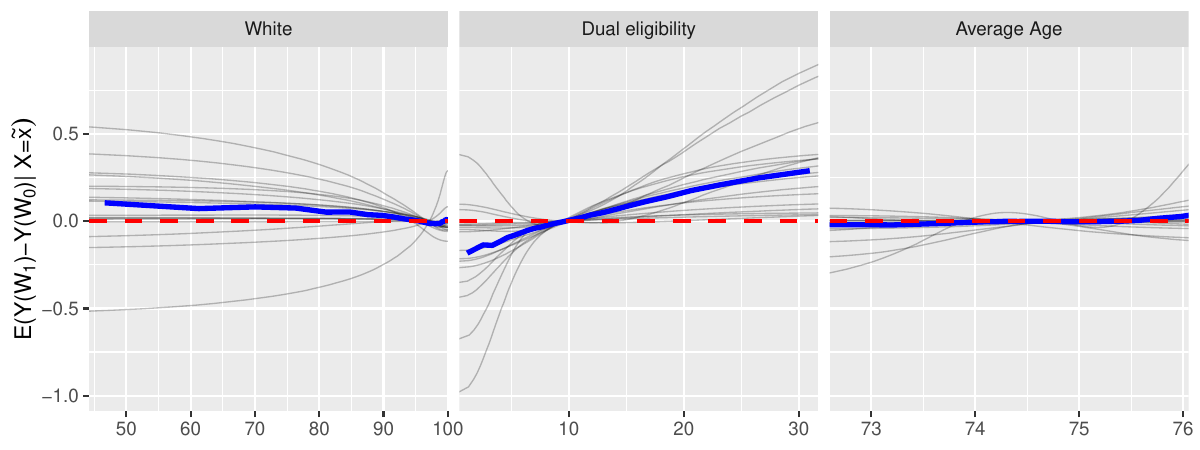}
    \caption{Estimates of conditional average treatment effects on mortality rates of increasing all pollutants from the first quartiles to the third quartiles as a function of each covariate, $\E\{Y(\bw_1) - Y(\bw_0) | \bX = \widetilde\bx_{(j)}\}$ where $\widetilde{\bx}_{(j)} = (\widebar{x}_{1}, ..., x_j, ..., \widebar{x}_p)$. Each thin gray curve represents the posterior mean for one year and the thick blue curve represents the average of posterior means over the study years. The red dashed line represents no treatment effect heterogeneity.}
    \label{fig:data_hetero}
\end{figure}

\section{Discussion}

In this paper, we proposed a novel approach to analyze and summarize complex, heterogeneous effects of multivariate continuous exposures. We developed new estimands for this setting, including a treatment effect variable importance measure, which is tailored to multivariate exposure scenarios and provides interpretable quantities that simplify the heterogeneity and provide important epidemiological insights about the effects of air pollution mixtures. To facilitate identification of our estimands, and to allow different strengths of regularization for each component of the model, we proposed a separation of the outcome model regression function, and integrated a targeted smoothing method into our model to obtain smooth estimates of the degree of heterogeneity by each covariate. Our theoretical results, coupled with simulation studies, provide strong empirical and analytical support for the efficiency and validity of the proposed model. We used the model to gain new insights into the causal impact of multivariate air pollution mixtures on mortality rates in the Medicare population. Our model estimates a detrimental impact of air pollution, consistent with previous literature, and shows that this effect is more pronounced in zip-codes with lower socioeconomic status, fewer white individuals, and older individuals.

While this paper provides strong evidence of heterogeneous treatment effects for environmental mixtures, and provides users with novel methodology for the multivariate, continuous treatment setting, there are certain limitations with our approach. For one, our definition of the MTE-VIM requires the user to choose a reference exposure level $\bw_0$. While we expect most reasonable choices of $\bw_0$ to lead to similar values of the MTE-VIM as we have seen in Appendix \ref{sec:Appendix_w0}, results could be sensitive to this choice in certain cases. To avoid violations of the positivity assumption and reliance on model extrapolation, we recommend users select $\bw_0$ from a region of the exposure space with sufficient support in the data.

An additional limitation of our approach is that we made the assumption that treatment effect heterogeneity was additive in the covariates, which allowed us to write $h(\bX,\bW)=\sumjo[p] h_j(X_j,\bW)$. We find this to be a reasonable assumption in most cases, and one that facilitates estimation in a difficult setting, but more complex forms of heterogeneity may not be captured well by this additive structure. For cases where treatment effect heterogeneity is expected to be non-additive in covariates, one could completely remove the assumption by fitting $h(\bX,\bW)$ with all covariates and exposures at once, or relax the assumption to allow for multiplicative interactions between two covariates and exposures. However, this would require adjustments to the proposed identification strategy and the interpretation of each functional component, and it is unclear how much this approach would impact estimation efficiency, particularly in cases with many covariates, as in our real data analysis.

Lastly, there are many future research directions to consider. For one, our focus has been on outcomes measured on a continuous scale, but a natural and important extension would be to allow for binary or count outcomes. The extension to binary outcomes with a probit or logistic link is straightforward using latent variables within our MCMC algorithm \citep{albert1993bayesian, polson2013bayesian}. While less trivial, BART has recently been extended to count outcomes \citep{murray2021log}, and future work could combine our proposed methodology within BART models for count outcomes. One other important direction would be to couple these estimation strategies with sensitivity analysis to unmeasured confounding, a common concern in observational studies such as this one. The recent literature has developed advancements in sensitivity analysis in multiple exposure settings for estimating average treatment effects \citep{zheng2021bayesian}, and these ideas could be extended to heterogeneous treatment effects seen here. Another interesting direction would be to combine the approach developed in this manuscript with methodology for optimal policy estimation, in order to provide environmental regulators increased information about the best practices for reducing pollution in the future in order to obtain the largest public health benefit from a proposed reduction in air pollution levels.

%%%%%%%%%%%%%%%%%%%%%%%%%%%%%%%%%%%%%%%%%%%%%%
%% Appendix---Please move all appendices to %%
%% a Supplementary file.                    %%
%%%%%%%%%%%%%%%%%%%%%%%%%%%%%%%%%%%%%%%%%%%%%%
%% Support information, if any,             %%
%% should be provided in the                %%
%% Acknowledgements section.                %%
%%%%%%%%%%%%%%%%%%%%%%%%%%%%%%%%%%%%%%%%%%%%%%
%\begin{acks}[Acknowledgments]
% The authors would like to thank ...
%\end{acks}
%%%%%%%%%%%%%%%%%%%%%%%%%%%%%%%%%%%%%%%%%%%%%%
%% Funding information, if any,             %%
%% should be provided in the                %%
%% funding section.                         %%
%%%%%%%%%%%%%%%%%%%%%%%%%%%%%%%%%%%%%%%%%%%%%%
\begin{funding}
Research described in this article was conducted under contract to the Health Effects Institute (HEI), an organization jointly funded by the United States Environmental Protection Agency (EPA) (Assistance Award No. CR-83590201) and certain motor vehicle and engine manufacturers. The contents of this article do not necessarily reflect the views of HEI, or its sponsors, nor do they necessarily reflect the views and policies of the EPA or motor vehicle and engine manufacturers. The computations in this paper were run on the FASRC FASSE cluster supported by the FAS Division of Science Research Computing Group at Harvard University.
Heejun Shin, Danielle Braun and Kezia Irene were also funded by the following grants from the National Institute of Health; R01AG066793, RF1AG074372, RF1AG071024, R01ES030616, R01ES034373, RF1AG080948, R01ES034021
\end{funding}

%%%%%%%%%%%%%%%%%%%%%%%%%%%%%%%%%%%%%%%%%%%%%%
%% Supplementary Material, including data   %%
%% sets and code, should be provided in     %%
%% {supplement} environment with title      %%
%% and short description. It cannot be      %%
%% available exclusively as external link.  %%
%% All Supplementary Material must be       %%
%% available to the reader on Project       %%
%% Euclid with the published article.       %%
%%%%%%%%%%%%%%%%%%%%%%%%%%%%%%%%%%%%%%%%%%%%%%
\begin{supplement}
\stitle{Proofs and Additional Results}
\sdescription{\cite{shin2024supplement} contains the proofs of the posterior contraction rates and identifiability of the proposed model, and additional simulation and analysis results.}
\end{supplement}
\begin{supplement}
\stitle{R Package}
\sdescription{The SepBART R package is available to implement the proposed approach. The most up-to-date version is available at \url{https://github.com/hshin111/SepBART}.}
\end{supplement}

%%%%%%%%%%%%%%%%%%%%%%%%%%%%%%%%%%%%%%%%%%%%%%%%%%%%%%%%%%%%%
%%                  The Bibliography                       %%
%%                                                         %%
%%  imsart-nameyear.bst  will be used to                   %%
%%  create a .BBL file for submission.                     %%
%%                                                         %%
%%  Note that the displayed Bibliography will not          %%
%%  necessarily be rendered by Latex exactly as specified  %%
%%  in the online Instructions for Authors.                %%
%%                                                         %%
%%  MR numbers will be added by VTeX.                      %%
%%                                                         %%
%%  Use \cite{...} to cite references in text.             %%
%%                                                         %%
%%%%%%%%%%%%%%%%%%%%%%%%%%%%%%%%%%%%%%%%%%%%%%%%%%%%%%%%%%%%%

%% or include bibliography directly:
% \begin{thebibliography}{}
% \bibitem[\protect\citeauthoryear{???}{???}]{b1}
% \end{thebibliography}
\appendix

\section{Proof of Theorem \ref{Thm: Convergence}}\label{sec: proof}
Before showing the proof of the main result, we must first show the conditions on the SoftBART prior distribution that are required, which we refer to as Condition P. Specifically, we need that
\begin{enumerate}[label= (P\arabic*)]
    \item There exists positive constants $(C_{M1}, C_{M2})$ such that the prior on the number of trees $T$ in the ensemble is $\Pi(T = t) = C_{M1} \exp\{-C_{M2} t \log t\}$.
    \item A single bandwidth parameter $\tau_t \equiv \tau$ is used and its prior satisfies $\Pi(\tau \ge x) \le C_{\tau1} \exp(-x^{C_{\tau 2}})$ and $\Pi(\tau^{-1} \ge x) \le C_{\tau 3} \exp(-x^{C_{\tau 4}})$ for some positive constants $C_{\tau 1}, \ldots, C_{\tau 4}$ for all sufficiently large $x$, with $C_{\tau 2}, C_{\tau 4} < 1$. Moreover, the density of $\tau^{-1}$ satisfies $\pi_{\tau^{-1}}(x) \ge C_{\tau 5} e^{-C_{\tau 6} x}$ for large enough $x$ and some positive constants $C_{\tau 5}$ and $C_{\tau 6}$.
    \item The prior on the splitting proportions is $s \sim \ndirichlet(a / r^\xi, \ldots, a / r^\xi)$ for some $\xi > 1$ and $a < 0$, where $r$ is the dimension of the input space (i.e., $p$ for $f$, $q$ for $g$, and $p + q$ for $h$).
    \item The $\mu_{t\ell}$'s are \iid from a density $\pi_{\mu}(\mu)$ such that $\pi_{\mu}(\mu) \ge C_{\mu 1} e^{-C_{\mu 2} |\mu|}$ for some coefficients $C_{\mu 1}, C_{\mu 2}$. Additionally, there exists constants $C_{\mu 3}, C_{\mu 4}$ such that $\Pi(|\mu_{t \ell}| \ge x) \le C_{\mu 3} \exp(-x^{C_{\mu 4}})$ for all $x$.
    \item Let $D_t$ denote the depth of tree $\mathcal{T}_t$. Then $\Pi(D_t = k) > 0$ for all $k = 0, 1, \ldots, 2d$ and $\Pi(D_t > d_0) = 0$ for $d_0 \ge 2d$ where $d$ represents the maximum number of covariates on which the function depends.
    \item The gating function $\psi: \reals \to [0,1]$ of the SoftBART prior is such that $\sup_x |\psi'(x)| < \infty$ and the function $\rho(x) = \psi(x) \{1 - \psi(x)\}$ is such that $\int \rho(x) \ dx > 0$, $\int |x|^m \, \rho(x) < \infty$ for all integers $m \ge 0$, and $\rho(x)$ can be analytically extended to some strip $\{z : |\Im(z)| \le U\}$ in the complex plane.
\end{enumerate}
All conditions above can be achieved by using the default SoftBART prior described in \cite{linero2018bayesian}. We consider the concentration of the posterior with respect to the norm $\|\cdot\|_n$ defined by $\|\mu\|_n^2 = \frac{1}{n} \sum_{i=1}^n \mu(\bX_i, \bW_i)^2$. We will first verify the following conditions (Condition C) for the prior and ground truth $\mu_0$ when Condition A and Condition P hold to establish the result.
\begin{enumerate}[label= (C\arabic*)]
    \item The prior satisfies $\Pi(\|\mu - \mu_0\|_n \le \epsilon_n) \ge e^{-C n \epsilon_n^2}$ for sufficiently large $n$ and constant $C$ independent of $(n, p, q)$.
    \item There exists a \textit{sieve} $\sF_1 \subseteq \sF_2 \subseteq \cdots$ such that
$$
  \log N(\sF_n, \bar \epsilon_n, \|\cdot\|_n) \le C_N n \bar \epsilon_n^2,
$$
for sufficiently large $n$, where $C_N$ is a constant independent of $(n, p, q)$ and $\bar \epsilon_n$ is a constant multiple of $\epsilon_n$. Here, $N(\sF, \delta, \|\cdot\|_n)$ denotes the $\delta$-covering number, which is the smallest number of $\delta$-balls required to cover $\sF$ with respect to $\|\cdot\|_n$.
    \item For the same $\sF_n$ from C2 and same constant $C$ from C1, we have
$$
  \Pi(\mu \notin \sF_n) \le e^{-(C + 4)n \epsilon_n^2},
$$
for sufficiently large $n$.
\end{enumerate}
Note that it suffices to consider $\|\cdot\|_\infty$ in place of $\|\cdot\|_n$ for each of (C1)--(C3) because trivially $\|\mu\|_n \le \|\mu\|_\infty$. Following from the results of \cite{ghosal2007convergence} (see, in particular, Theorem 1, Theorem 4, Section 7.7, and the comments following Theorem 2), verifying Condition C is sufficient to prove Theorem \ref{Thm: Convergence}.

Assuming Conditions A and P, we can state the following useful results, which follow from the proof of Theorem 2 of \cite{linero2018bayesian}.
\begin{proposition}\label{prop: pc}
Suppose that Condition A holds, and independent priors are specified for $(f, g, h)$ such that Condition P holds for each. Then for any $c > 0$ there exist constants and $C_f$ independent of $(n, p, q)$ such that for sufficiently large $n$ we have
$$
  \Pi(\|f - f_0\|_\infty \le c \epsilon_{nf}) \ge e^{-C _fn \epsilon_{nf}^2}.
$$
Similarly, there exist $C_g$ and $C_h$ such that $\Pi(\|g - g_0\|_\infty \le c \epsilon_{ng}) \ge e^{-C_g n \epsilon_{ng}^2}$ and $\Pi(\|h - h_0\| \le c \epsilon_{nh}) \ge e^{-C_h n \epsilon_{nh}^2}$. 
\end{proposition}
The following lemma now verifies Condition (C1).
\begin{lemma}
    Suppose that Condition (A1) holds with $\sigma^2$ fixed and that the conditions of Proposition \ref{prop: pc} also hold. Then there exists a constant $C$ such that for sufficiently large $n$
$$
  \Pi(\|\mu - \mu_0\|_n \le \epsilon_n) \ge e^{-C n \epsilon_n^2}
$$
where $\epsilon_n = \max(\epsilon_{nf}, \epsilon_{ng}, \epsilon_{nf})$.
\end{lemma}
\begin{proof}
    First, we note that by the triangle inequality and prior independence that
$$
\begin{aligned}
  \Pi(\|\mu - \mu_0\|_\infty \le \epsilon_n)
  &\ge
  \Pi(\|f - f_0\|_\infty \le \epsilon_n / 3) 
  \Pi(\|g - g_0\|_\infty \le \epsilon_n / 3) 
  \Pi(\|h - h_0\|_\infty \le \epsilon_n / 3).
\end{aligned}
$$
Applying Proposition \ref{prop: pc} with $c = 1/3$ and the definition of $\epsilon_n$ this gives
$$
  \Pi(\|\mu - \mu_0\|_\infty \le \epsilon_n)
  \ge e^{-(C_f + C_g + C_h) n \epsilon_n^2} = e^{-C n \epsilon_n^2}
$$
where $C = C_f + C_g + C_h$.
\end{proof}
Next, we prove the existence of an appropriate sieve to establish (C2)--(C3). We use the following result, which is proven in the supplementary material of \cite{linero2018bayesian}, to do this.
\begin{lemma}
    Suppose that Condition P holds. For universal constants $(D, E)$ (depending only on the prior) and for sufficiently large $(T, d, U, \sigma_2)$ and sufficiently small $(\sigma_1, \delta)$, there exists a set $\sG_f$ with the following properties:

1. \textnormal{\textbf{Covering entropy control:}} $\log N(\sG_f, D \delta, \|\cdot\|_\infty) \le d \log p + 3 T 2^{d_0} \log\left(\frac{d \sigma^2_2  T 2^{d_0} U}{\delta}\right)$

2. \textnormal{\textbf{Complement probability bound:}} If $\Pi_f$ satisfies Condition P then
$$
\begin{aligned}
  \Pi_\mu(\sG_f^c) 
  &\le C'_{M1} \exp\{-C'_{M2} T \log(T)\} + 
    \\&\qquad 2^{d_0} T \left[\exp\{-E d \log p\} + C_{\mu 1} \exp\{-U^{C_{\mu 2}}\}\right] + 
    \\&\qquad 2 C_{\tau_1} \, \exp\{-\sigma_1^{-C_{\tau 2}}\} + 2 C_{\tau_3} \exp\{-\sigma_2^{C_{\tau_4}}\},
\end{aligned}
$$

Similarly, sets $\sG_g$ and $\sG_h$ exist with analogous statements holding for $\Pi_g$ and $\Pi_h$. 
\end{lemma}
Using this lemma, we can construct an appropriate sieve as follows. First, define the Cartesian product $\sF = \sG_f \times \sG_g \times \sG_h$ for appropriately large/small values of $(T, d, U, \sigma_2, \sigma_1, \epsilon)$. Then we immediately have
$$
  \log N(\sF, D\delta, \|\cdot\|_\infty) \le
  \log N(\sG_f, D\delta, \|\cdot\|_\infty) + 
  \log N(\sG_g, D\delta, \|\cdot\|_\infty) + 
  \log N(\sG_h, D\delta, \|\cdot\|_\infty).
$$
Following \cite{li2022adaptive}, for a large constant $\kappa$ to be selected later, take
$$
  \sigma_1^{-C_{\tau 2}} = \sigma_2^{C_{\tau 4}} = U^{C_{\mu 2}} = 
  \kappa n \epsilon_n^2,
$$
$d = \lfloor \kappa n \epsilon_n^2 / \log(p + q + 1) \rfloor$ and $T = \lfloor \kappa n \epsilon_n^2 / \log n\rfloor$. By substituting these constants into the covering entropy bound, we can obtain the inequality
$$
  \log N(\sF, D \epsilon_n, \|\cdot\|_\infty)
  \le \kappa' n \epsilon_n^2
$$
for an appropriately chosen $\kappa' > \kappa$ and large enough $n$, which verifies (C2). Additionally, by plugging these choices into the complementary probability bound, we can also make
$$
  \Pi(\sF^c) \le 
  \Pi(f \in \sG_f^c) + \Pi(g \in \sG_g^c) + \Pi(h \in \sG_h^c)
  \le 
  \exp\{-(C + 4) n \epsilon_n^2\}
$$
for sufficiently large $n$ by taking $\kappa$ sufficiently large as desired. (C1)--(C3) are verified.

\section{Proof of identifiability}\label{sec: identifiability proof}

We prove that each function component in Equation \ref{eq: model} is identifiable by incorporating the shifting strategy introduced in Section \ref{sec: identification}.
\begin{proposition}
    Suppose $\mathcal{S}_{\bx_0,\bw_0}$ denotes a space of a tuple $\big(c, f(\cdot), g(\cdot), h(\cdot)\big)$ that satisfies $f(\bx_0)=g(\bw_0)=0$ and $h\{\bx_0,\bw\}=h\{\bx,\bw_0\}=0$ for all $\bx\in\reals^p$ and $\bw\in\reals^q$. If
    \begin{enumerate}
        \item $\big\{\big(c, f(\cdot), g(\cdot), h(\cdot)\big), \big(\widetilde{c}, \widetilde{f}(\cdot), \widetilde{g}(\cdot), \widetilde{h}(\cdot)\big)\big\}\subset\mathcal{S}$, and
        \item $c + f(\bx) + g(\bw) + h(\bx,\bw)=\widetilde{c} + \widetilde{f}(\bx) + \widetilde{g}(\bw) + \widetilde{h}(\bx,\bw)$ for all $\bx\in \reals^p$ and $\bw\in\reals^q$,
    \end{enumerate}
    then $\big(c, f(\cdot), g(\cdot), h(\cdot, \cdot)\big)=\big(\widetilde{c}, \widetilde{f}(\cdot), \widetilde{g}(\cdot), \widetilde{h}(\cdot, \cdot)\big)$.
\end{proposition}
\begin{proof}
Let $\bx=\bx_0$ and $\bw=\bw_0$. Since both tuples belong to $\mathcal{S}_{\bx_0,\bw_0}$, the second condition becomes $c=\widetilde{c}$. Next, we fix $\bx=\bx_0$ while allowing $\bw$ to vary. Then, we have $c + g(\bw) =\widetilde{c}  + \widetilde{g}(\bw)\Leftrightarrow c + g(\bw) =c  + \widetilde{g}(\bw) \Leftrightarrow g(\bw)= \widetilde{g}(\bw)$ for all $\bw\in\reals^q$. Analogously, setting $\bw=\bw_0$ yields $f(\bx)= \widetilde{f}(\bx)$ for all $\bx\in\reals^p$. Combining the above results produces $c + f(\bx) + g(\bw)=\widetilde{c} + \widetilde{f}(\bx) + \widetilde{g}(\bw)$ for all $\bx\in \reals^p$ and $\bw\in\reals^q$. Hence, it follows that $h(\bx,\bw)=\widetilde{h}(\bx,\bw)$ for all $\bx\in \reals^p$ and $\bw\in\reals^q$. Setting $\bx_0=\E(\bX)$ and $\bw_0=\E(\bW)$ concludes the proof.

\end{proof}

\section{Additional simulation studies}\label{sec:AdditionalSim}

In this section, we present and discuss additional simulation results. First we present results showing the performance of our blocking scheme for calculating MTE-VIMs, and then present results in a simulation study where the additive treatment effect assumption does not hold. 

\subsection{Examining performance of blocking scheme}\label{sec:AppendixBlocking}

In this section, we present additional simulation results to support the use of the blocking scheme in calculating variable importance measures, as described in Section \ref{sec: estimation}. We employ the same simulation scenario as that in Section \ref{sec: simulation}, where strong interactions exist between covariates and exposures. Recall that the true parameter of interest $(\psi_1,\psi_2,\psi_3,\psi_4,\psi_5)=(0.722,0.278,0,0,0)$.

We first compare three methods for estimating the parameter: 1) using a randomly chosen 100 observations from the original data set (Small); 2) using all observations (Full); 3) partitioning observations into ten blocks and averaging the resulting ten estimates (Block). Table \ref{tab: Blocking RMSE} presents the relative RMSE compared to the blocking scheme averaged over 300 data replicates. A relative RMSE greater than 1 indicates a larger RMSE than the blocking scheme. We also present the histograms of the point estimates of $\psi_2$ for each method in Figure \ref{fig: block compare}. Using just one subset of the original observations results in a significantly larger RMSE as it has a larger variability in estimating the parameter, which is expected since it doesn't incorporate all data points. The blocking scheme, which simply averages estimates from ten small subsets, shows a comparable estimation performance to the ``Full" method which requires significantly more intensive computation.

\begin{table}[h]
    \centering
    \begin{tabular}{c||c|c|c}
        Estimand & Small & Full & Block\\
        \hline
        $\psi_1$ & 1.29 & 0.99 & 1 \\
        $\psi_2$ & 1.66 & 0.99 & 1 \\
        $\psi_3$ & 1.23 & 1.04 & 1 \\
        $\psi_4$ & 1.59 & 1.17 & 1 \\
        $\psi_5$ & 1.66 & 0.99 & 1
    \end{tabular}
    \caption{Relative RMSE of estimating MTE-VIM.}
    \label{tab: Blocking RMSE}
\end{table}
\begin{figure}[h]
    \centering
    \includegraphics[width=4 in]{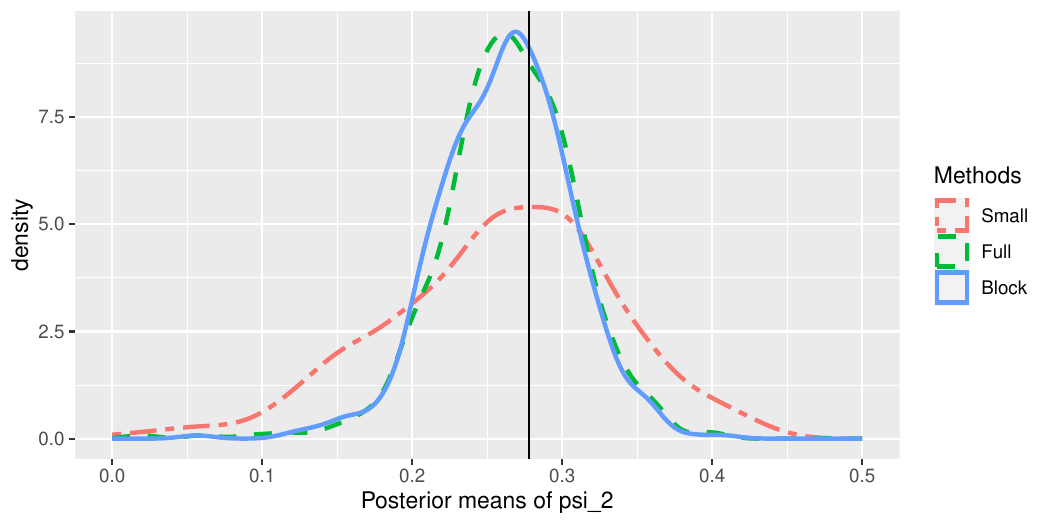}
    \caption{Histograms of 500 posterior means of the variable importance measures, $\psi_2$. The vertical line represents the true value of 0.278.}
    \label{fig: block compare}
\end{figure}

For a detailed comparison of the estimators with and without the blocking scheme, Figure \ref{fig:blockfull} shows the distributions of posterior means of the MTE-VIM and the rejection rates of the null hypotheses for both methods. This comparison is conducted under the same simulation settings as those in Section \ref{sec: simulation}, with the sample size reduced to 1000 due to the computational constraints of running the full method with larger sample sizes. There are effectively no differences between the two methods, supporting the validity of using the blocking scheme in our analysis.

\begin{figure}
    \centering
    \includegraphics[width=0.9\linewidth]{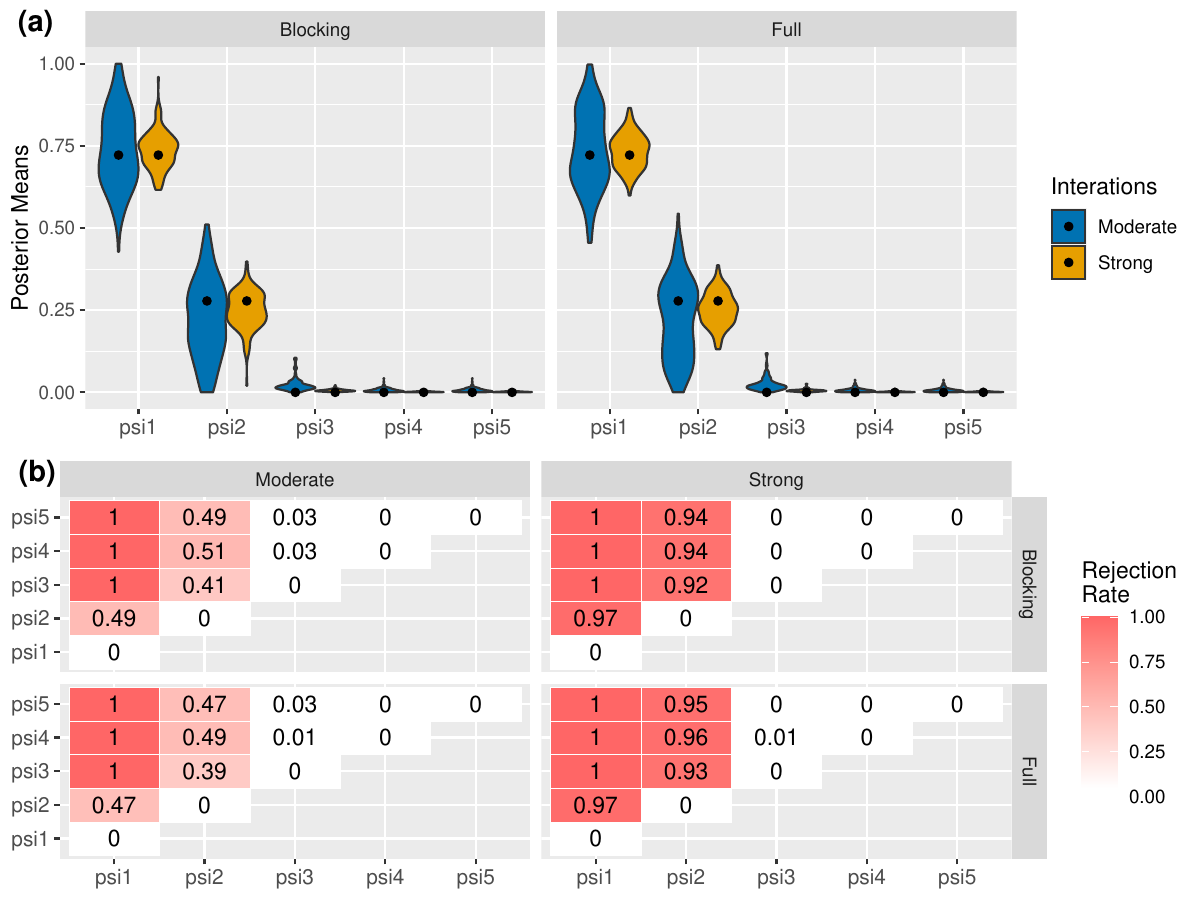}
    \caption{Comparing the performance of MTE-VIM estimation methods with and without the blocking scheme. (a) Violin plots of 300 posterior means of the variable importance measures. The black dots represent the true variable importance measures; (b) Rejection rates for each test where $H_0: \psi_i=\psi_j$ with the level $\alpha=0.05$.}
    \label{fig:blockfull}
\end{figure}

\subsection{Non-additive treatment effect heterogeneity}\label{sec:AppendixNonAdditive}

Our model presented in Section \ref{sec: estimation} assumes that the interactions between covariates and exposures are additive in covariates. While this assumption is generally mild in practice, it is important to assess how robust our model is to violations of the additive interaction assumption. To do this, we consider two non-additive interaction scenarios in which two covariates jointly modify the treatment effect while all other data-generating processes remain the same as those described in Section \ref{sec: simulation}:
\begin{align*}
    &\text{Violation 1: } h(\bX,\bW) = 0.5\cos(X_1)\cos(X_2)g(W)&\\
    &\text{Violation 2: } h(\bX,\bW) = 0.35\indct{X_1<1}\indct{X_2<1}g(W)&
\end{align*}
For both violation scenarios, the true MTE-VIM is $(\psi_1,\psi_2,\psi_3,\psi_4,\psi_5)=(0.5,0.5,0,0,0)$.

\begin{figure}
    \centering
    \includegraphics[width=0.75\linewidth]{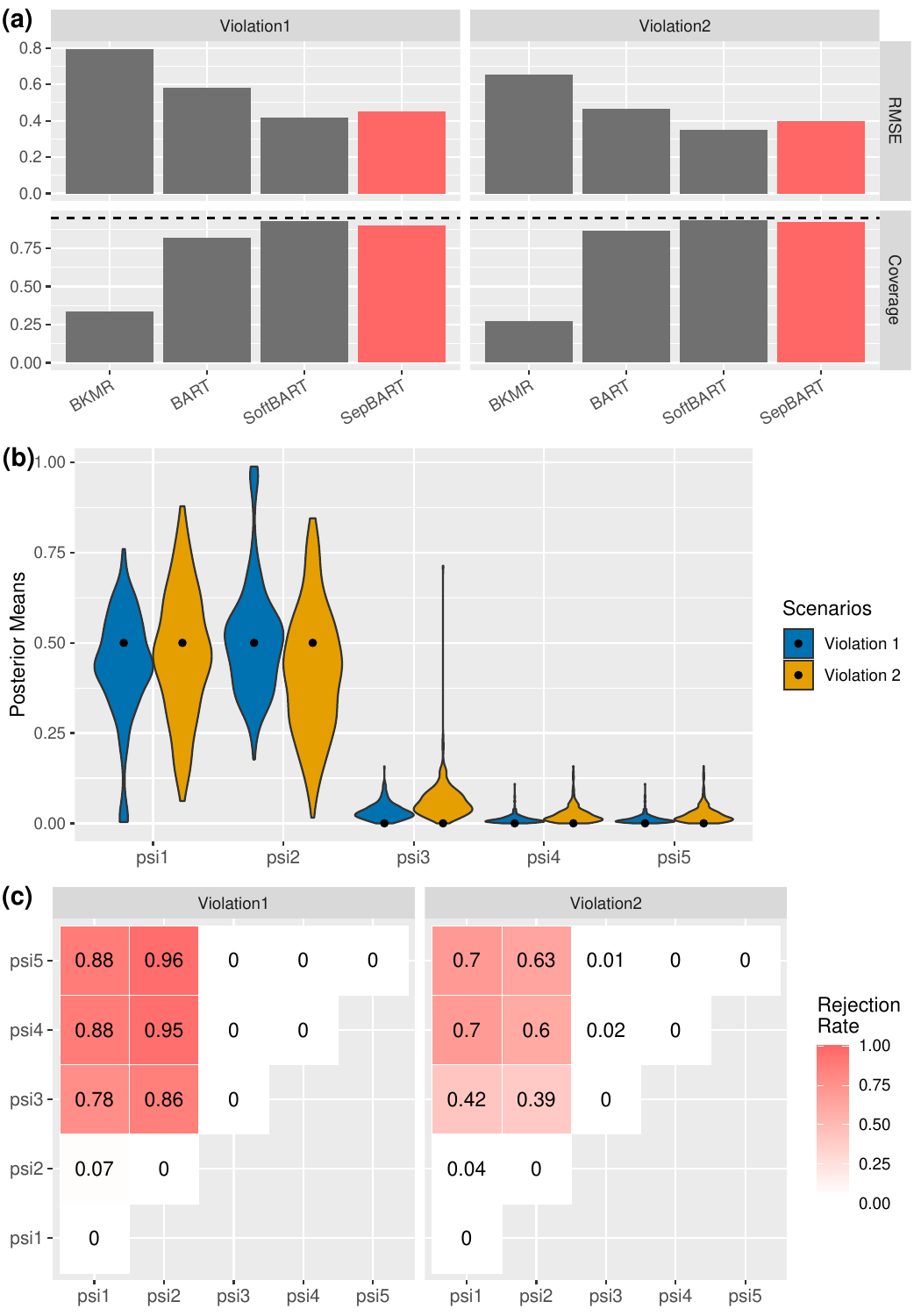}
    \caption{Comparing the performance of the proposed model when the additive interaction assumption is violated. (a) RMSE and coverage for CATE described in Section \ref{sec: simulation}; (b) Violin plots of 300 posterior means of the variable importance measures. The black dots represent the true variable importance measures; (c) Rejection rates for each test where $H_0: \psi_i=\psi_j$ with the level $\alpha=0.05$.}
    \label{fig:violation}
\end{figure}
Panel (a) in Figure \ref{fig:violation} shows the performance of CATE estimation for our model compared to competitors. We observe that even when the additive interaction assumption is violated, our model achieves a similar RMSE to the existing BART models that do not make this assumption, and performs better than BKMR, which does not allow for treatment effect heterogeneity. Regarding the estimation of MTE-VIM, Panel (b) shows that our model detects the true MTE-VIM well, though with a heavier tail compared to when the interaction is additive in covariates. In Panel (c), we observe that for both violation scenarios, the rejection rates are well controlled. Overall, the findings in this section suggest that although our model relies on the additive interaction structure, it performs reasonably well even when the assumption is moderately violated.

\section{Additional data analysis results}\label{sec:AppendixAnalysis}

In this section, we present and discuss additional results from the Medicare analysis. These include convergence diagnostics, examining robustness to the choice of $\boldsymbol{w}_0$, examining grouped variable importance metrics instead of single covariate MTE-VIMs, and an assessment of the positivity assumption in our analysis. 

\subsection{Convergence of the model}\label{sec:Appendix_convergence}

It is important to check the convergence of the MCMC algorithm for our data analysis, which involves many covariates and exposures with complex relationships. In Section \ref{sec:MedicareAnalysis}, we ran three chains with 10,000 MCMC iterations, discarding the first 4,000 as a burn-in, and thinning every twelfth sample. Figure \ref{fig:ATETrace} presents a trace plot of ATE for one year of data analysis, showing that all three chains overlap without any noticeable pattern. Similar patterns were observed for all other years of analysis. Another way to check the convergence of the model is to examine the potential scale reduction factor (PSRF) \citep{gelman1992inference}. The PSRF should be close to 1 if MCMC chains converged to the target distribution and a PSRF value below 1.2 is considered as evidence of convergence \citep{brooks1998general}. In our analysis, the PSRF for estimating the ATE is 1.013. Since the calculation of the ATE involves all components of our model, this result suggests good overall model convergence.

\begin{figure}[h]
    \centering
    \includegraphics[width=0.8\linewidth]{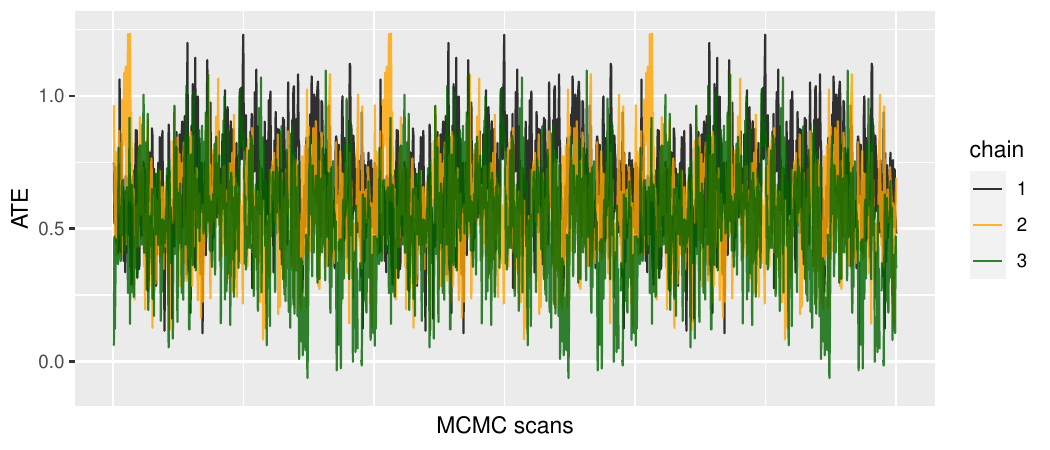}
    \caption{Trace plot of the ATE for one year of the Medicare analysis.}
    \label{fig:ATETrace}
\end{figure}

We also examined the PSRF for MTE-VIM, and the median PSRF across all years of analysis is shown for each covariate in Table \ref{tab:PSRF_MTEVIM}. While we observed an overall increase in PSRF for MTE-VIM compared to ATE, the PSRF is well-controlled for all covariates except for dual eligibility to Medicaid. This is expected, given the inherent challenges in isolating the effect of a single covariate in complex nonparametric models, particularly in the presence of highly correlated exposures and covariates, as is the case in our Medicare analysis. Unlike the overall conditional mean function, $\E\{Y(\bw)\vert \bX\}$, which is easier to estimate, MTE-VIMs depend heavily on the individual $h_j(\cdot)$ functions for each covariate, which are more difficult to identify in settings with strong correlations among covariates and exposures. The posterior distribution of MTE-VIMs likely has multiple modes, which can result in similar values of $\E\{Y(\bw)\vert \bX\}$, and therefore the ATE, but slightly different values of the MTE-VIMs.

\begin{table}[ht]
\centering
\begin{tabular}{r|r}
\hline
Covariates &PSRF \\
  \hline
Dual eligibility & 1.33  \\ 
  Income  & 1.09 \\ 
  Population density & 1.11 \\ 
  Owner occupied & 1.06\\ 
  Female  & 1.11 \\ 
  White &1.14  \\ 
  Low Education &1.10  \\ 
  Average Age  &1.11\\ 
  BMI &1.12\\ 
  Smoke  &1.09\\ 
  Temperature &1.07\\ 
  Humidity  &1.11\\ 
   \hline
\end{tabular}
    \caption{The median across all years of analysis of the potential scale reduction factor of each MTE-VIM.}
    \label{tab:PSRF_MTEVIM}
\end{table}

\subsection{Robustness to choices of $w_0$}\label{sec:Appendix_w0}

The proposed MTE-VIM requires the user to select a reference exposure level, $\bw_0$. While one can generally expect that the resulting MTE-VIM does not vary significantly within a reasonable range of $\bw_0$ values, we assess this by conducting the same data analysis as in Section \ref{sec:MedicareAnalysis} with two alternative choices of $\bw_0$. Specifically, our original choice of $\bw_0$ is the 25th quantile of each exposure, and in this additional analysis, we also consider the 50th and 75th quantiles. Figure \ref{fig:diffw0} displays the average MTE-VIM for each choice of $\bw_0$, and the results show that the MTE-VIM remains consistent across different choices of $\bw_0$ for all covariates, as expected.
\begin{figure}[h]
    \centering
    \includegraphics[width=0.95\linewidth]{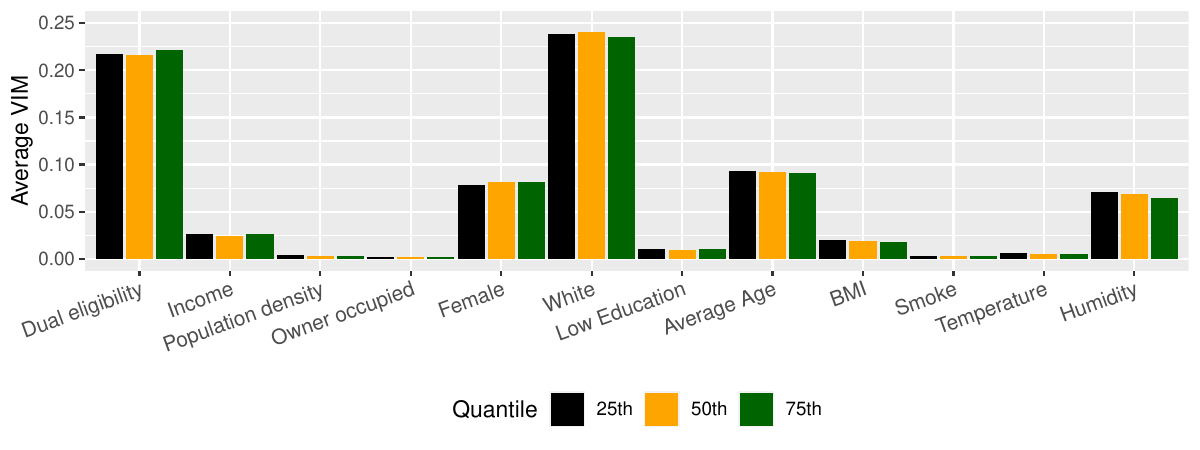}
    \caption{The average MTE-VIM with alternative $\bw_0$ values. The 25th quantile is used in the manuscript.}
    \label{fig:diffw0}
\end{figure}

\subsection{Incorporating the grouped variable importance measures}\label{sec:Appendix_GVIM}

Because MTE-VIM, $\{\psi_j\}_{j=1}^p$ could be affected by the correlated covariates as discussed in Section \ref{sec: estimands}, one could consider a grouped MTE-VIM, $\psi_s$, where $s$ is a subset of $\{1,...,p\}$. To illustrate how the grouped MTE-VIM can be applied in data analysis, we categorize the covariates into three groups: 1) Socioeconomic variables (dual eligibility, median household income, population density, percent owner-occupied housing, percent female, percent low education), 2) Risk factors (average age, mean BMI, smoking rate), and 3) Climate variables (temperature, humidity). The results are presented in Figure \ref{fig:GVIM}. We found that the socioeconomic variable group achieves a high MTE-VIM across all years, with an average of 0.7, suggesting that the majority of treatment effect heterogeneity can be explained solely by the socioeconomic variables. This finding is not surprising, as the socioeconomic variable group contains the largest number of covariates, some of which individually achieve high MTE-VIM values, such as race and dual eligibility to Medicaid.
\begin{figure}[h]
    \centering
    \includegraphics[width=0.8\linewidth]{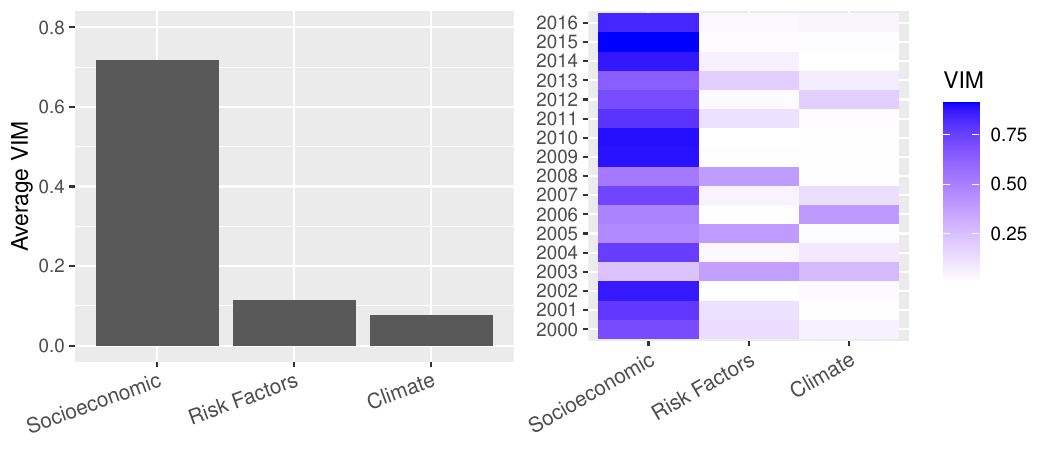}
    \caption{The average across years of the grouped MTE-VIM (left), and the grouped MTE-VIM for each year (right).}
    \label{fig:GVIM}
\end{figure}

\subsection{Positivity and a trimmed estimator}\label{sec:Appendix_positivity}

In our analysis in Section \ref{sec:MedicareAnalysis}, we examine $\E\{Y(\bw_1)-Y(\bw_0)\}$, where $\bw_0=(w_{01},...,w_{0q})$ and $\bw_1=(w_{11},...,w_{1q})$ represent vectors of exposures corresponding to the first and third quartiles for each exposure, respectively. One common concern when analyzing multivariate, continuous exposures is the difficulty of assessing the positivity assumption. In this section, we marginally check the positivity assumption, specifically at our choice of $\bw_0$ and $\bw_1$.

\begin{table}[t]
    \centering
    \begin{tabular}{c|c|c|c|c|c|c||c}
        \pmtpf & Ozone  &  Elemental Carbon   & Organic Carbon &  Ammonium &  Nitrate  & Sulfate & All Exposures \\
        \hline
        0.93 &0.95 &0.87& 0.91 &0.93& 0.97& 0.88 & 0.74
    \end{tabular}
    \caption{The proportion of observations satisfying the univariate positivity criterion for each exposure. The last column indicates the proportion of observations satisfying all univariate criteria simultaneously.}
    \label{tab:positivity_marginal}
\end{table}
We begin by checking the positivity assumption for each exposure separately. For the $j$-th exposure, we evaluate the following two probabilities: $p_{0j} = P(W_j \in \mathcal{R}(w_{0j}) | \bX = \bx)$ and $p_{1j} = P(W_j \in \mathcal{R}(w_{1j}) | \bX = \bx)$, where $\mathcal{R}(w)$ represents an interval containing $w$. These probabilities, $p_{0j}$ and $p_{1j}$, correspond to the probabilities of being near the exposure levels of interest $w_{0j}$ and $w_{1j}$, conditional on the covariates, and are expected to be non-zero if the positivity assumption holds. We consider the positivity assumption for the $j$-th exposure to be satisfied if $\min(p_{0j}, p_{1j}) > \delta$, where $\delta$ is a chosen threshold. Table \ref{tab:positivity_marginal} shows the proportion of observations that meet this univariate positivity criterion, averaged over the study years. For this, we set $\mathcal{R}(w_{0j}) = [r_{j,0.15}, r_{j,0.35}]$ and $\mathcal{R}(w_{1j}) = [r_{j,0.65}, r_{j,0.85}]$, where $r_{j,a}$ denotes the $100 \times a$-th quantile of the $j$-th exposure, and $\delta = 0.01$. We assumed that the marginal distribution of each exposure conditional on the covariates follows a normal distribution, with the mean and variance estimated using a third-degree polynomial linear regression. We observed that every exposure had at least 88\% of the observations satisfied the univariate criteria, indicating that most observations have sufficient probability of being exposed to both the neighborhood of $w_{0j}$ and $w_{1j}$ for each exposure. Extending this further, we can check what proportion of observations had this univariate positivity criteria hold for all exposures simultaneously. This means that $p_{zj} > \delta$ for all $j = 1, ..., q$ and for both $z = 0$ and $z = 1$. We found that 74\% of our observations met this more restrictive threshold, suggesting that the positivity assumption is reasonably satisfied in our analysis.

Although assessing positivity using the marginal distribution of each exposure, as described in the previous paragraph, provides some indication of the plausibility of the positivity assumption for multivariate exposures, it does not fully confirm the assumption, which involves the joint distribution of exposures given covariates. To address this more complicated issue, we define $p_{0} = P(\bW \in \mathcal{R}(\bw_0) | \bX = \bx)$ and $p_{1} = P(\bW \in \mathcal{R}(\bw_1) | \bX = \bx)$, where $\mathcal{R}(\bw)$ is a hyperrectangle around $\bw$, and examine $\min(p_{0}, p_{1})$. Determining a reasonable threshold for positivity violations in the multivariate case is not straightforward, as this issue has not been explored fully in the causal inference literature for multivariate, continuous treatments. Despite this, we can assess robustness to this positivity assumption by constructing an alternative, trimmed estimand that targets the ATE among observations with larger values of $\min(p_{0}, p_{1})$, which is more robust to positivity violations. These types of trimming estimators have been explore extensively for binary treatments, \citep{crump2009dealing, yang2018asymptotic} and have recently been used for univariate, continuous treatments \citep{branson2023causal}. We extend these ideas to the multivariate setting seen here.  Specifically, we define the trimmed ATE as:
$$\E\{Y(\bw_1)-Y(\bw_0)\vert \min(p_{0}, p_{1}) > \delta_T \}$$
where $\delta_T$ is set to the median of $\min(p_{0}, p_{1})$ across the entire sample. In other words, we focus on the observations with the highest probabilities of having exposures near $\bw_0$ and $\bw_1$, as these are least likely to violate the positivity assumption. The trimmed ATE can be estimated using the following plug-in estimator:
$$\dfrac{\sum_{i=1}^n\indct{\min(\widehat{p}_{i0}, \widehat{p}_{i1}) > \delta_T}\widehat{\E}\{Y(\bw_1)-Y(\bw_0)\vert \bX=\bx_i\}}{\sum_{i=1}^n\indct{\min(\widehat{p}_{i0}, \widehat{p}_{i1}) > \delta_T}}.$$
The estimated original and trimmed ATE are presented in Figure \ref{fig:Trimmed}. While the trimmed estimates are slightly smaller (by approximately 0.05) than the original ATE estimates, the overall similarity between the two suggests that the results obtained are reasonable and not affected largely by positivity violations.

\begin{figure}
    \centering
    \includegraphics[width=0.8\linewidth]{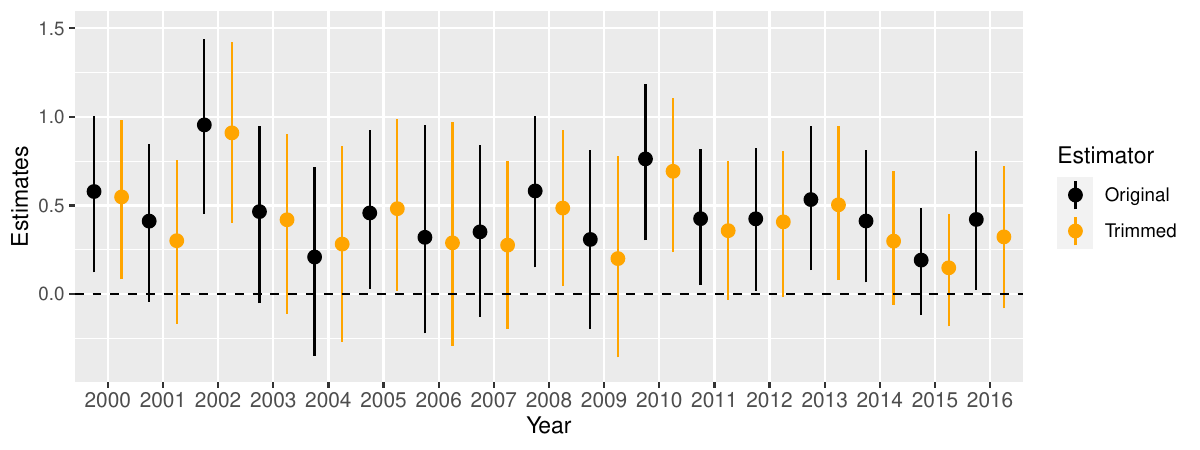}
    \caption{Posterior means and corresponding 95\% credible intervals for the original ATE and for the trimmed ATE when every pollutant simultaneously increases from its yearly first quartile to the third quartile. The dashed line represents no average treatment effect.}
    \label{fig:Trimmed}
\end{figure}

%% if your bibliography is in bibtex format, uncomment commands:
\bibliographystyle{imsart-nameyear} % Style BST file
\bibliography{ref}       % Bibliography file (usually '*.bib')
\end{document}